 \documentclass[authoryear,preprint,12pt]{elsarticle}

\usepackage{amssymb}
\usepackage{array}
\newcolumntype{C}[1]{>{\centering\let\newline\\\arraybackslash\hspace{0pt}}m{#1}}

 \usepackage{amsthm}
 \usepackage{amsmath}
 \newtheorem{lemma}{Lemma}
 \newtheorem{definition}{Definition}
 \newtheorem{theorem}{Theorem}
\newtheorem{proposition}{Proposition}

\usepackage{multirow}
\usepackage{caption}
\usepackage{subcaption}

\usepackage{bm}
\newcommand\independent{\protect\mathpalette{\protect\independenT}{\perp}}
\def\independenT#1#2{\mathrel{\rlap{$#1#2$}\mkern2mu{#1#2}}}

\usepackage{pgf, tikz} 
\usetikzlibrary{arrows,automata,fit}
\usetikzlibrary{shapes}
 
\newcommand{\xx}{2.2}
\newcommand{\yy}{2}

\usepackage{booktabs}
\usepackage{hyperref}

\journal{}

\begin{document}

\begin{frontmatter}



\title{The diameter of a stochastic matrix: A new measure for sensitivity analysis in Bayesian networks}

 \author[label1]{Manuele Leonelli}
 \affiliation[label1]{organization={School of Science and Technology, IE University, Madrid}, country = {Spain}}

  \author[label2,label3]{Jim Q. Smith}
  \author[label2]{Sophia K. Wright}
 \affiliation[label2]{organization={Department of Statistics, University of Warwick, Coventry}, country = {UK}}
  \affiliation[label3]{organization={The Alan Turing Institute, London}, country = {UK}}
  

\begin{abstract}
Bayesian networks are one of the most widely used classes of probabilistic models for risk management and decision support because of their interpretability and flexibility in including heterogeneous pieces of information. 
In any applied modelling, it is critical to assess how robust the inferences on certain target variables are to changes in the model. 
In Bayesian networks, these analyses fall under the umbrella of sensitivity analysis, which is most commonly carried out by quantifying dissimilarities using Kullback-Leibler information measures.
In this paper, we argue that robustness methods based instead on the familiar total variation distance provide simple and more valuable bounds on robustness to misspecification, which are both formally justifiable and transparent.
We introduce a novel measure of dependence in conditional probability tables called the diameter to derive such bounds. This measure quantifies the strength of dependence between a variable and its parents.
We demonstrate how such formal robustness considerations can be embedded in building a Bayesian network.
\end{abstract}



\end{frontmatter}


\section{Introduction}

Bayesian networks (BNs) \citep[e.g][]{koller2009probabilistic,Smith2010} are now a well-established and widely-used AI modelling approach for a wide range of risk management applications. They support decision makers by providing an intuitive graphical framework to reason about the dependence of various risk factors and an efficient platform to perform inferential queries, scenario and sensitivity analyses. Their use as a decision support tool in business and OR has been increasing over the years, including case studies in project management \citep{van2020dependent}, supply chain \citep{garvey2015analytical}, marketing \citep{hosseini2021decision}, and logistics \citep{qazi2022adoption}, among others. 

BNs are defined by two components: a directed acyclic graph (DAG) where each node is a variable of interest and edges represent the, possibly causal, relationship between them; a conditional probability table (CPT) for each node of the DAG reporting the probability distribution of the associated variable conditional on its parents. BNs are highly interpretable due to their graphical nature, representing the probabilistic relationships between variables, making it easy for users to understand and trace the influence of one variable on another. With explainability now recognized as critical for the use of AI in applied research \citep{rudin2019stop}, including in OR \citep{de2023explainable}, BNs stand out by providing transparent and intuitive explanations, thereby enhancing trust and clarity in decision-making processes.

The underlying DAG and the associated CPTs can be learned from data using machine learning algorithms or elicited using experts' opinions and knowledge. There is now a vast amount of algorithms to learn BN from data \citep[e.g.][for a review]{scutari2019learns}, as well as various software implementing such routines \citep[most notably the \texttt{bnlearn} R package,][]{Scutari2010}. Furthermore, there are now protocols to guide the construction of BNs from the input of experts \citep{french2011aggregating,renooij2001probability,werner2017expert,wilkerson2021customized}.

No matter how the BN was constructed, in any real-world modelling with BNs, it is critical to assess the importance of various risk factors and evaluate the overall robustness of the model to misspecifications of its inputs. Such a step is usually referred to as \emph{sensitivity analysis} \citep{borgonovo2017sensitivity,razavi2021future} and is a fundamental ingredient of applied mathematical modelling in any area of science \citep{borgonovo2016sensitivity,saltelli2000sensitivity,saltelli2021sensitivity}. A variety of sensitivity methods for BNs have been introduced \citep[see e.g.][]{ballester2023yodo,chan2002numbers,Rohmer2020,van2007sensitivity} and implemented in various pieces of software, \citep[e.g. the \texttt{bnmonitor} R package,][]{leonelli2023sensitivity}.

Most sensitivity methods in BNs either use the KL-divergence or the Chan-Darwiche distance to quantify the dissimilarity between two BNs with different input parameters \citep{chan2005distance,leonelli2022geometric,renooij2014co}. However, these measures depend very heavily on the accurate specification of small probabilities since they are specified in log-probabilities and ratio of probabilities, respectively. In many applied scenarios, mainly when BNs are used as a decision support tool, the misspecification of improbable events has only a small impact on the required outputs of a decision analysis. For this reason, in this paper, we propose the use of a much more conventional distance measure \citep[widely used in probability theory and stochastic analysis,][]{sason2016f}, namely the \textit{total variation distance}. Although it is often difficult to derive explicit formulae for the impacts of deviation in variation, it is nevertheless straightforward to tightly bound such deviations, as we demonstrate below.

In particular, we introduce a novel measure, called \textit{diameter}, which summarizes the information carried by a CPT of the BN in variation distance. By amount of information, we broadly mean strength of dependence between a variable and its parents in the DAG, although the exact intuition varies depending on the type of analysis conducted, as we illustrate in the following sections. The diameter will be essential for constructing tight bounds to the effect of the model's perturbations derived from a suite of sensitivity studies.

Unlike almost all sensitivity methods available to BNs that require a complete specification of the whole BN, the methods proposed here can rely on a partial specification of the CPTs. Despite the critical role of sensitivity analysis \textit{during} the knowledge engineering process of constructing and eliciting a BN \citep{coupe2000sensitivity,Laskey2000}, methods based on a partial model specification are scarce \citep[see e.g.][for two exceptions]{Albrecht2014,boneh2006matilda}.

An implementation of the developed routines is freely available via the \texttt{bnmonitor} R package \citep{leonelli2023sensitivity}. Proofs of the main results are collated in the appendix.

\section{Bayesian networks and sensitivity analysis}

\subsection{Bayesian networks}

Let $G=([n],E)$ be a DAG with vertex set $[n]=\{1,\dots,n\}$ and edge set $E$. Let $\bm{X}=(X_i)_{i\in[n]}$ be categorical random variables with joint probability mass function (pmf) $p$ and sample space $\mathbb{X}=\times_{i\in[n]}\mathbb{X}_i$. For $A\subset [n]$, we let $\bm{X}_A=(X_i)_{i\in A}$ and $\bm{x}_A=(x_i)_{i\in A}$ where $\bm{x}_A\in\mathbb{X}_A=\times_{i\in A}\mathbb{X}_i$. We say that $p$ is Markov to $G$ if, for $\bm{x}\in\mathbb{X}$, 
\begin{equation}
    \label{eq:factorization}
p(\bm{x})=\prod_{i\in[n]}p(x_i \mid \bm{x}_{\Pi_i}),
\end{equation}
where $\Pi_k$ is the parent set of $k$ in $G$ and $p(x_k \vert \bm{x}_{\Pi_k})$ is a shorthand for $p(X_k=x_k \vert \bm{X}_{\Pi_k} = \bm{x}_{\Pi_k})$.  The factorization in Equation (\ref{eq:factorization}) is equivalent to a set of conditional independences as formalized by the \textit{local Markov property}. In a BN with DAG $G$ it holds that, for every $i\in [n]$,
\[
X_i\independent X_{ND_i} | X_{\Pi_i},
\]
where $ND_i$ is the set of non-descendants of $i$ in $G$. For categorical variables, each term of the factorization in Equation (\ref{eq:factorization}) is taken from the CPT of the relevant vertex. The CPT is a \textit{stochastic matrix} (i.e. non-negative rows summing to one) where each row correspond to a combination of the parents.

The DAG associated to a BN provides an intuitive overview of the relationships between variables of interest. However, it also provides a framework to assess if any generic conditional independence holds for a specific subset of the variables via the so-called \emph{d-separation} criterion \citep[see e.g.][]{koller2009probabilistic}. Efficient propagation of probabilities and evidence to compute any (conditional) probability involving a specific subset of the variables can be carried out exploiting the underlying DAG structure, which we discuss in Section \ref{sec:computing} below. Next, we introduce two examples to illustrate BNs.

\subsection{Two examples}

\begin{figure}
     \centering
     \begin{subfigure}[b]{0.48\textwidth}
         \centering
         \includegraphics[width=\textwidth]{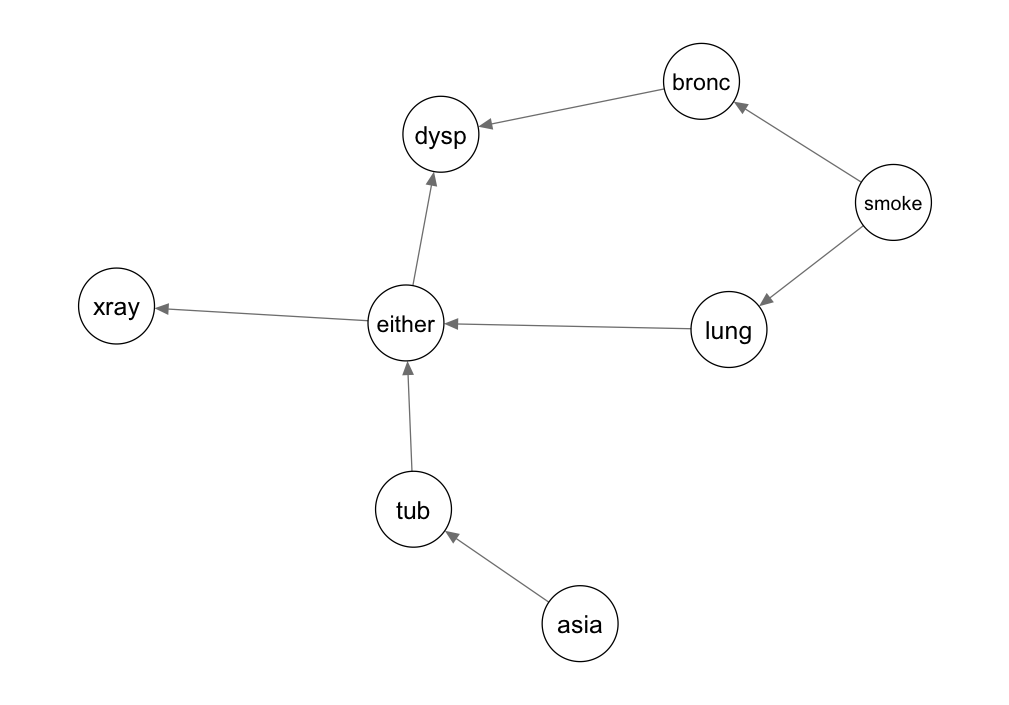}
         \caption{The \texttt{asia} BN}
         \label{fig:asia}
     \end{subfigure}
     \hfill
     \begin{subfigure}[b]{0.48\textwidth}
         \centering
         \includegraphics[width=\textwidth]{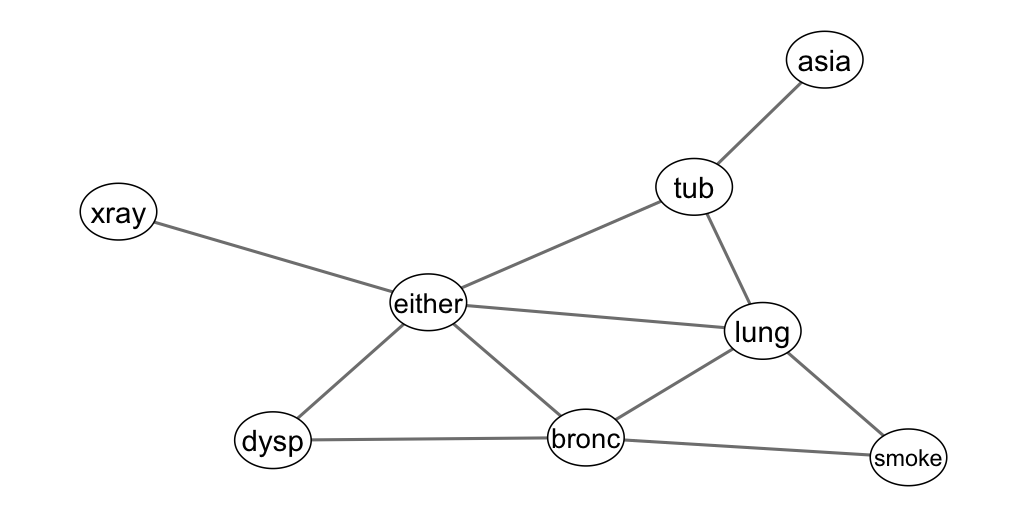}
         \caption{The triangulated \texttt{asia} BN.}
         \label{fig:asia1}
     \end{subfigure}
        \caption{Illustration of the triangulation of a BN.}
        \label{fig:three graphs}
\end{figure}

We first consider the well-known \texttt{asia} BN reported in Figure \ref{fig:asia}, representing the possible causes of shortness of breath, also called dyspnea (\texttt{dysp}). A visit to Asia (\texttt{asia}) may cause tuberculosis (\texttt{tub}), while smoking may cause bronchitis (\texttt{bronc}) and lung cancer (\texttt{lung}). The vertex \texttt{either} is an or gate, which is equal to yes if either \texttt{lung} or \texttt{tub} is equal to yes. The output of a chest x-ray (\texttt{xray}) only depends on the presence of lung cancer or tuberculosis, while dyspnea depends on all three diseases (bronchitis, lung cancer, tuberculosis). The factorization of the pmf for this BN is 
\begin{multline*}
p(\texttt{xray}|\texttt{either})p(\texttt{dysp}|\texttt{bronc},\texttt{either})p(\texttt{either}|\texttt{lung},\texttt{tub})\\p(\texttt{bronc}|\texttt{smoke})p(\texttt{lung}|\texttt{smoke})p(\texttt{tub}|\texttt{asia})p(\texttt{smoke})p(\texttt{asia}).
\end{multline*}

\begin{table}[]
    \centering
    \scalebox{0.6}{
    \begin{tabular}{|C{3.5em}|C{3.5em}|C{5.5em}|C{5.5em}|}
    \hline
    \multirow{2}{*}{\texttt{bronc}} & \multirow{2}{*}{\texttt{either}} & \multicolumn{2}{c|}{$P(\texttt{dysp}|\texttt{either},\texttt{bronc})$} \\
    \cline{3-4}
         & & yes & no \\
         \hline
    yes & yes &  0.9 & 0.1\\
    \hline
    yes & no &  0.8 & 0.2\\
    \hline 
    no & yes  & 0.7 & 0.3\\
    \hline 
    no & no & 0.1 & 0.9 \\
    \hline
    \end{tabular}
    \hspace{1cm}
        \begin{tabular}{|C{3.5em}|C{3.5em}|C{5.5em}|C{5.5em}|}
    \hline
    \multirow{2}{*}{\texttt{lung}} & \multirow{2}{*}{\texttt{tub}} & \multicolumn{2}{c|}{$P(\texttt{either}|\texttt{tub},\texttt{lung})$} \\
    \cline{3-4}
         & & yes & no \\
         \hline
    yes & yes &  1.0 & 0.0\\
    \hline
    yes & no &  1.0 & 0.0\\
    \hline 
    no & yes  & 1.0 & 0.0\\
    \hline 
    no & no & 0.0 & 1.0 \\
    \hline
    \end{tabular}
    }
    \caption{CPTs from the \texttt{asia} BN associated to the vertices \texttt{dysp} (left) and \texttt{either} (right).}
    \label{tab:CPTs}
\end{table}

Table \ref{tab:CPTs} reports two CPTs from the \texttt{asia} BN. Each row represents a combination of the parent variables and the associated numeric entries are non-negative and sum to one. Therefore, they are stochastic matrices. The variable \texttt{either} is deterministically defined by its parents.

The second example is a BN introduced in \citet{varando2024staged}, learned over data from the 2012 Italian enterprise innovation survey collected by ISTAT \citep{ISTAT2015}, the Italian national statistical institute.  The survey reports information about medium-sized Italian companies and their involvement with innovation in  2010-2012. The analysis aims to assess which factors related to innovation are connected with changes in the company revenue. The considered variables are reported in Table \ref{tab:istat} and details about data pre-processing can be found in \citet{varando2024staged}. Figure \ref{fig:istat} reports the learned BN using the \texttt{tabu} algorithm of the \texttt{bnlearn} R package \citep{Scutari2010}, comprising of 15 vertices and 38 edges. The output variable (\texttt{GROWTH}) is directly affected by the number of employees of the company (\texttt{EMP12}) and whether or not the company carried out other innovation activities in 2010-2012 (\texttt{INPD}). Its CPT is reported in Table \ref{table:growth}.

\begin{table}[]
    \centering
    \scalebox{0.65}{
    \begin{tabular}{|c|c|c|}
    \toprule
    Name & Explanation & Levels \\
    \midrule
        \texttt{GP} (P)& Belongs to an industrial group & Yes/No  \\
        \texttt{LARMAR} (L) & Main market & Regional/National/International\\
         \texttt{INPDGD} (D) & Product innovation 2010-2012 & Yes/No \\
         \texttt{INPDSV} (V)& Service innovation 2010-2012 & Yes/No\\
         \texttt{INPD} (N)& Other innovations 2010-2012 & Yes/No \\
         \texttt{INABA} (A)& Abandoned innovation 2008-2010 & Yes/No\\
         \texttt{INONG} (I)& Ongoing innovation from 2008-2010 & Yes/No\\
         \texttt{CO} (C)& Cooperation agreements for innovation & Yes/No\\
         \texttt{ORG} (O)& New organization practices & Yes/No\\
         \texttt{MKT} (M)& New marketing practices & Yes/No\\
         \texttt{PUB} (B)& Contracts with public institutions & Yes/No\\
         \texttt{EMP12} (2)& Number of employees in 2012 & 10-49/50-249/$>$250\\
         \texttt{EMPUD} (E)& Employees with degree & 0\%/1-10\%/$>$10\%\\
         \texttt{RR} (R)&  Research \& development & Yes/No\\
         \texttt{GROWTH} (G)& Increased revenue 2012/2010 & Yes/No\\
         \bottomrule
    \end{tabular}
    }
\caption{Variables from the 2012 ISTAT enterprise innovation survey.}
    \label{tab:istat}
\end{table}

\begin{figure}
    \centering
    \includegraphics[scale=0.25]{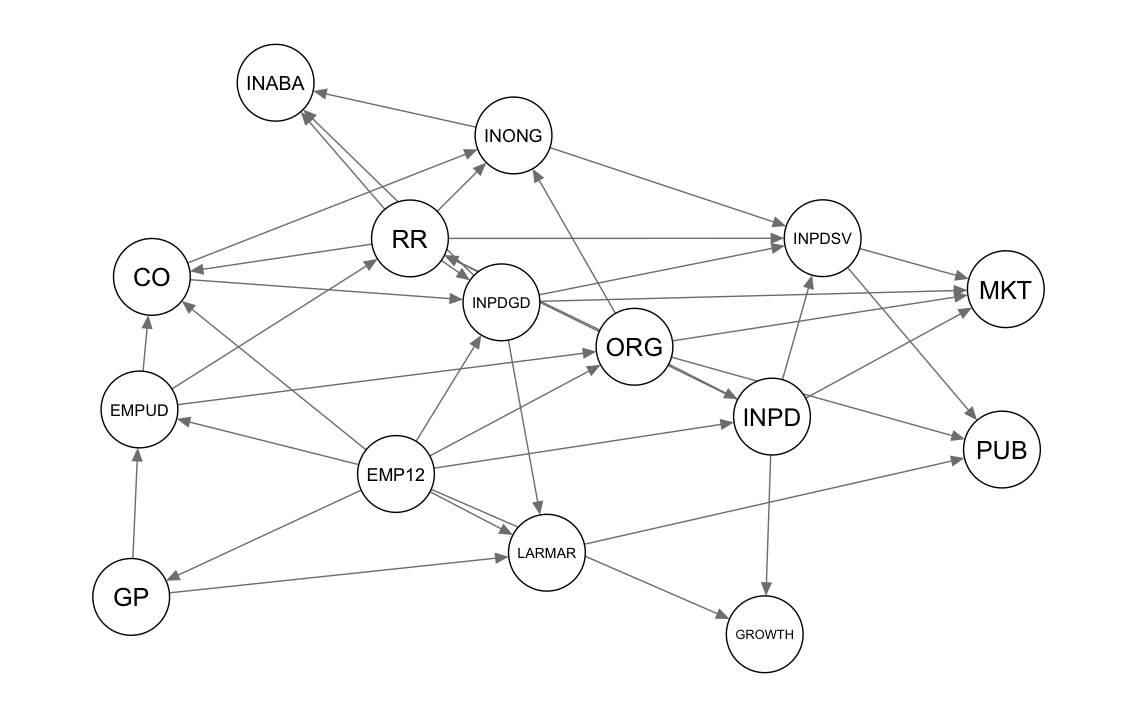}
    \caption{The \texttt{istat} BN learned over the 2012 ISTAT enterprise innovation data.}
    \label{fig:istat}
\end{figure}

\begin{table}[]
    \centering
    \scalebox{0.6}{
    \begin{tabular}{|C{3.5em}|C{3.5em}|C{5.5em}|C{5.5em}|}
    \hline
    \multirow{2}{*}{\texttt{EMP12}} & \multirow{2}{*}{\texttt{INPD}} & \multicolumn{2}{c|}{$P(\texttt{GROWTH}|\texttt{INPD},\texttt{EMP12})$} \\
    \cline{3-4}
         & & yes & no \\
         \hline
    10-49 & yes &  0.563 & 0.437\\
    \hline
    10-49 & no &  0.469 & 0.531\\
    \hline 
    50-249 & yes  & 0.608 & 0.392\\
    \hline 
    50-249 & no & 0.557 & 0.443 \\
    \hline
        $>250$ & yes  & 0.636 & 0.364\\
    \hline 
     $>250$ & no & 0.590 & 0.410 \\
    \hline
    \end{tabular}
    }
    \caption{CPT associated to the vertex \texttt{GROWTH} in the \texttt{istat} BN in Figure \ref{fig:istat}.}
    \label{table:growth}
\end{table}

\subsection{Computing probabilities in BNs}
\label{sec:computing}

Probabilistic inference in BNs is known to be NP-hard \citep{cooper1990computational}. However, algorithms that use the underlying DAG to speed up computations have been defined. One of the most famous algorithms transforms the original DAG into a  \textit{junction tree} \citep[see e.g.][]{koller2009probabilistic}. We will show in Section \ref{sec:error} that the junction tree can also be used for sensitivity investigations \citep[also used in][]{kjaerulff2000making}.

The junction tree algorithm first transforms  the original DAG into an undirected graph, by first applying \textit{moralization} (the addition of an edge between any two parents of the same vertex not joined by an edge), then dropping the directionality of the edges and lastly triangulating the graph (adding undirected edges until the resulting graph is such that every cycle of length strictly greater than 3 possesses a chord, that is, an edge joining two nonconsecutive vertices of the cycle). The result of this process over the \texttt{asia} BN is shown in Figure \ref{fig:asia1}.

The \textit{cliques} of the triangulated graph, its maximal fully connected subgraphs, $C_1,\dots,C_m$, can be totally ordered starting from any clique including a root of the original graph. Let $S_i=C_i\cap\cup_{j=1}^{i-1}C_j$ be the \textit{separator} of $C_i$ from the preceding cliques. The cliques can always be ordered to respect the \textit{running intersection property}, meaning that there is at least one $j<i$ such that $S_i\subset C_j$ for any $i\in[m]\setminus\{1\}$. This implies that the result of intersecting a clique with all previous cliques is contained within one or more earlier cliques. Any clique ordering and choice of separator containment can be depicted by a so-called \textit{junction tree}: an undirected tree graph with vertices $C_1,\dots,C_m$ and an undirected edge between $C_i$ and $C_j$ if $C_i$ is the chosen clique respecting $S_j\subset C_i$. The factorization in Equation (\ref{eq:factorization}) can be equivalently written in terms of the cliques as
\[
p(\bm{x})=\frac{\prod_{i\in[m]}p(\bm{x}_{C_i})}{\prod_{i\in[m]\setminus \{1\}}p(\bm{x}_{S_i})}= \prod_{i\in[m]\setminus\{1\}}p(\bm{x}_{C_i}|\bm{x}_{S_i})p(\bm{x}_{C_1}).
\]

For the \texttt{asia} BN, a possible ordering of its cliques is $C_1=\{\texttt{asia},\texttt{tub}\}$, $C_2=\{\texttt{tub},\texttt{lung},\texttt{either}\}$, $C_3=\{\texttt{either},\texttt{xray}\}$, $C_4=\{\texttt{bronc},\texttt{lung},\texttt{either}\}$, $C_5=\{\texttt{bronc},\texttt{either},\texttt{dysp}\}$, and $C_6=\{\texttt{smoke},\texttt{bronc},\texttt{lung}\}$, with separators $S_2=\{\texttt{tub}\}$, $S_3=\{\texttt{either}\}$, $S_4=\{\texttt{lung},\texttt{either}\}$, $S_5=\{\texttt{bronc},\texttt{either}\}$, and $S_6=\{\texttt{bronc},\texttt{lung}\}$. The resulting junction tree is reported in Figure \ref{fig:junction}, where, as customary, we labeled the edges with the separators.

\begin{figure}
    \centering
    \includegraphics[scale=0.2]{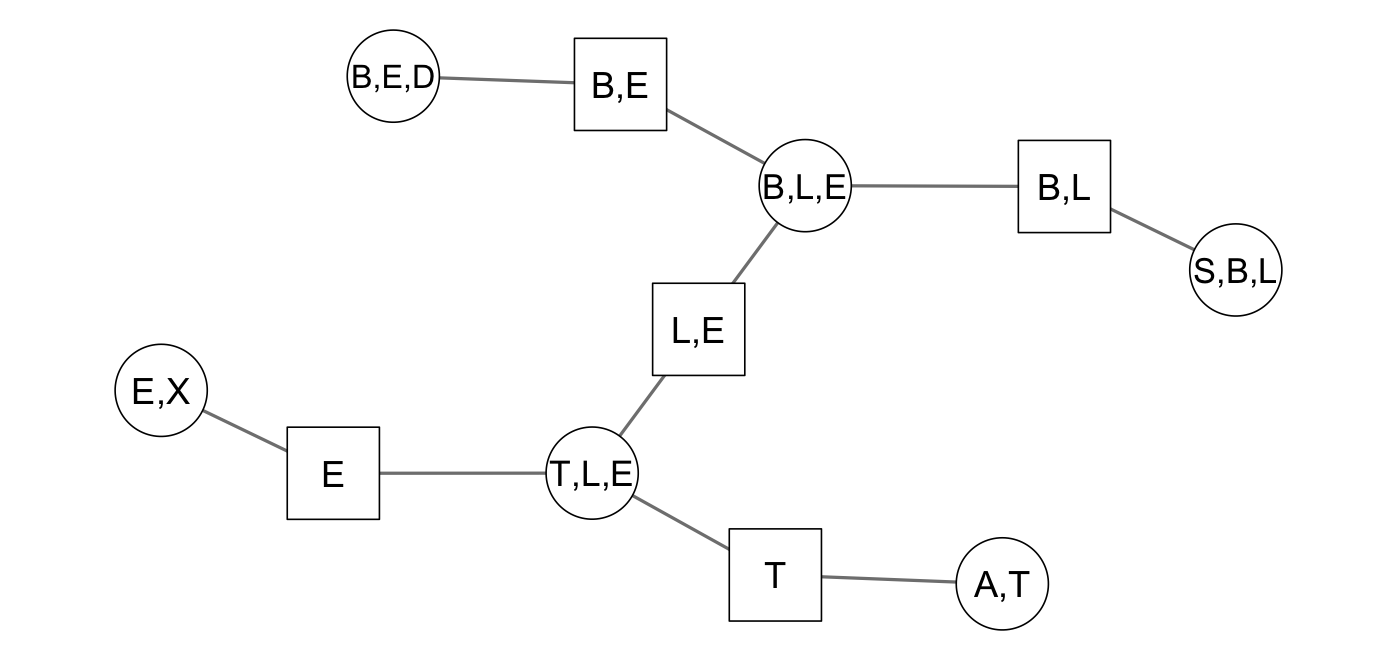}
    \caption{The junction tree representing the \texttt{asia} BN. Initial of variable names are used. Cliques are circled, while separators are squared. }
    \label{fig:junction}
\end{figure}

A straightforward consequence of the junction tree structure is that any two cliques, $C_i$ and $C_j$ say, are connected by a simple path. Assume that $C_j$ includes descendants of $C_i$ in the original DAG. This implies a sequence of separators in the unique path between $C_i$ and $C_j$, say $S_{i+1}$,\dots $S_j$. Letting $S_k^*=S_k\setminus \cup_{l=k+1}^{j}S_l $, for $k=i+1,\dots,j$, we have that
\begin{align}
p(\bm{x}_{C_i},\bm{x}_{C_j})&=\sum_{\bm{x}_{T_{ij}}\in\mathbb{X}_{T_{ij}}} p(\bm{x}_{C_i},\bm{x}_{C_j},\bm{x}_{T_{ij}})\nonumber\\
&=\sum_{\bm{x}_{T_{ij}}\in\mathbb{X}_{T_{ij}}} p(\bm{x}_{C_j}|\bm{x}_{S^*_j})p(\bm{x}_{S^*_j}|\bm{x}_{S^*_{j-1}})\cdots p(\bm{x}_{S^*_{i+2}}|\bm{x}_{S^*_{i+1}})p(\bm{x}_{C_i}),\label{eq:markov}
\end{align}
where $T_{ij}=\cup_{k=i+3}^{j-1}S^*_k\setminus \left\{C_j\cup C_i\right\}$. Equation (\ref{eq:markov}) can be thought of as the pmf of a ``donating" clique $C_i$ and a ``target" clique $C_j$, expressed in terms of a sequence of transitions in a non-homogeneous Markov chain. It means that standard results from non-homogeneous Markov chain theory can be used to measure the extent of the diminishing effect of information as it passes along the simple path from $C_i$ to $C_j$. In particular, it is well-known that variation distance in an ergodic, acyclic Markov chain contracts as information is propagated through the system \citep[e.g.][]{roberts2004general}. This observation will be critical for the developments of sensitivity bounds we introduce later on in this paper.

\subsection{Sensitivity analysis in BNs}

Sensitivity methods for BNs have been widely studied \citep[see][for the most comprehensive yet only partial review]{Rohmer2020}. Here, we provide an overview of BN models' most traditional sensitivity investigations. Notice that most methods require a complete specification of the whole model: its DAG and all entries of its CPTs.

The most widely used sensitivity analysis studies the effect of perturbations of the CPT entries on outputs of interest. Sensitivity functions mathematically model the relationship between inputs and outputs \citep{castillo1997sensitivity,Leonelli2017,van2007sensitivity}. Because of the computational complexity of deriving sensitivity functions for multiple parameter variations \citep{chan2004sensitivity,kwisthout2008computational}, in practice, most often only perturbations of one CPT entry at a time are performed, although recent methodological advances matching BNs to more flexible structures have made more complex investigations possible \citep{ballester2022you,salmani2023automatically}.

Another type of sensitivity investigation studies the overall effect of a node on an output. This is often quantified by the mutual information between the associated variables \citep{kjaerulff2008bayesian}. \citet{Albrecht2014} approached this problem by only considering the DAG of the BN, thus without requiring the CPTs to be defined, and introduced the \textit{distance weighted influence} between two variables $X_j$ and $X_i$ of a DAG. Let $S_{ji}$ be the set of active, simple trails from $j$ to $i$ \citep[see e.g.][]{koller2009probabilistic} and $w\in(0,1]$. The distance weighted influence of $X_j$ on $X_i$ is 
\begin{equation}
\label{eq:dwi}
DWI(X_j,X_i,w) = \sum_{s\in S_{ji}}w^{|s|},
\end{equation}
where $|s|$ is the length of the trail $s$. So $DWI$ measures how connected two vertices of a DAG are, where longer trails have a smaller contribution to the influence. If $w=1$, DWI simply counts the number of active, simple trails. Another way to measure the influence of a node on an output is to quantify how valuable it would be to observe the associated variable, the so-called sensitivity to evidence approach \citep{ballester2022computing,gomez2014sensitivity}.

Another class of sensitivity methods assesses the relevance of the edges of a BN. One standard way to do this for data-learned BNs is to use a non-parametric bootstrap approach and learn a BN for each replication: edges that do not appear frequently are deemed to have less strength \citep{scutari2013identifying}. Another possibility is to use sensitivity functions to quantify the effect of an edge removal \citep{Renooij2010}. Given the complexity of probabilistic inference, deleting edges having a small impact on the inferences made by the model \citep{choi2005bayesian,choi2006edge} is often desirable.


In Section \ref{sec:sens} we introduce new methods for the node relevance and edge strength problems mentioned above, but we also address investigations that have received less attention in the literature.

\section{Total variation and the diameter}

We next introduce our new metric, the diameter, based on the total variation distance between pmfs. 

\begin{definition}
Let $p$ and $p'$ be two pmfs over the same sample space $\mathbb{X}$. The \emph{total variation distance} between $p$ and $p'$ is 
\[
d_V(p,p')=\frac{1}{2}\sum_{x\in\mathbb{X}}|p(x)-p'(x)|.
\]
\end{definition}
As quantified through the absolute differences,  deviation in variation corresponds much more closely to the types of error we would envisage experiencing within either an elicitation exercise or through misestimation. 

We next define a measure of the overall total variation distance between the rows of a CPT. As already noticed, CPTs are stochastic matrices where each row is a pmf.

\begin{definition}
The (upper) \emph{diameter} of a $n\times m$ stochastic matrix $P$ with rows $p_1,\dots,p_n$, denoted as $d^+(P)$, is
\[
d^+(P)=\max_{i,j\in[ n]}d_V(p_i,p_j).
\]
The lower diameter  is
\[
d^-(P)=\min_{i,j\in[n]}d_V(p_i,p_j).
\]
\end{definition}

Our main focus is on the upper diameter which we henceforth simply refer to as the diameter, while the lower diameter will become relevant in Section \ref{sec:asy} only. The diameter is the largest variation distance between any two rows of a stochastic matrix. The use of the maximum distance is motivated by establishing bounds on the effect of probability perturbations and their propagation throughout the DAG.

We next characterize properties of the diameter and its relationship with conditional independence and marginalization, which will further shed light on its interpretation. We therefore now focus specifically on stochastic matrices representing CPTs. For a CPT we use the notation $P_{\cdot | \cdot}$, where the elements on the rhs of the subscript are the conditioning variables, while those on the lhs of the subscript are those of which the pmf is evaluated. For ease of interpretation, we use $X,Y,Z$ to denote generic categorical random variables (although all the results below can be written for categorical random vectors). 

\begin{proposition}
\label{prop:1}
   For two categorical random variables $X$ and $Y$ it holds
    \[
    d^+(P_{Y|X}) = 0 \Leftrightarrow Y\independent X.
    \]
\end{proposition}

The closer the diameter is to zero, the less dependent two variables are. To see this, it is easy to check that whenever some non-trivial function of $Y$ can be written as a deterministic function of $X$ then $d^+(P_{Y|X})=1$, its maximum value. So when changing the levels of $X$ has a minimum impact on the pmf of $Y$ then $d^+(P_{Y|X})\approx0$. Note that unless $P_{Y|X}$ is symmetric, $d^+(P_{Y|X})\neq d^+(P_{X|Y})$, in fact the difference between these can be arbitrarily close to 1 \citep{Wright2018}.

\begin{table}[]
    \centering
    \scalebox{0.65}{
    \begin{tabular}{cccccc}
    \toprule
       \texttt{tub}  & \texttt{lung} & \texttt{bronc} & \texttt{either} & \texttt{xray} & \texttt{dysp} \\
       \midrule
      0.04   & 0.09 & 0.30 & 1.00 & 0.93 & 0.80\\
      \bottomrule
    \end{tabular}}
    \caption{Diameter of the non-root nodes of the \texttt{asia} BN.}
    \label{tab:asia_diameter}
\end{table}

\begin{table}[]
    \centering
    \scalebox{0.65}{
    \begin{tabular}{cccccccccccccc}
    \toprule
       \texttt{GP}  & \texttt{LARMAR} & \texttt{INPDGD} & \texttt{INPDSV} & \texttt{INPD} & \texttt{INABA} & \texttt{INONG} & \texttt{CO} & \texttt{ORG} & \texttt{MKT} & \texttt{PUB} & \texttt{RR}&\texttt{EMPUD} & \texttt{GROWTH}\\
       \midrule
      0.666   & 0.611 & 0.857 & 0.859 & 0.648 & 0.500 & 0.761 & 0.484 & 0.546 & 0.697 & 0.344 & 0.941 & 0.521 & 0.167\\
      \bottomrule
    \end{tabular}}
    \caption{Diameter of the non-root nodes of the \texttt{istat} BN.}
    \label{tab:istat_diameter}
\end{table}

\begin{figure}
    \centering
    \includegraphics[scale=0.25]{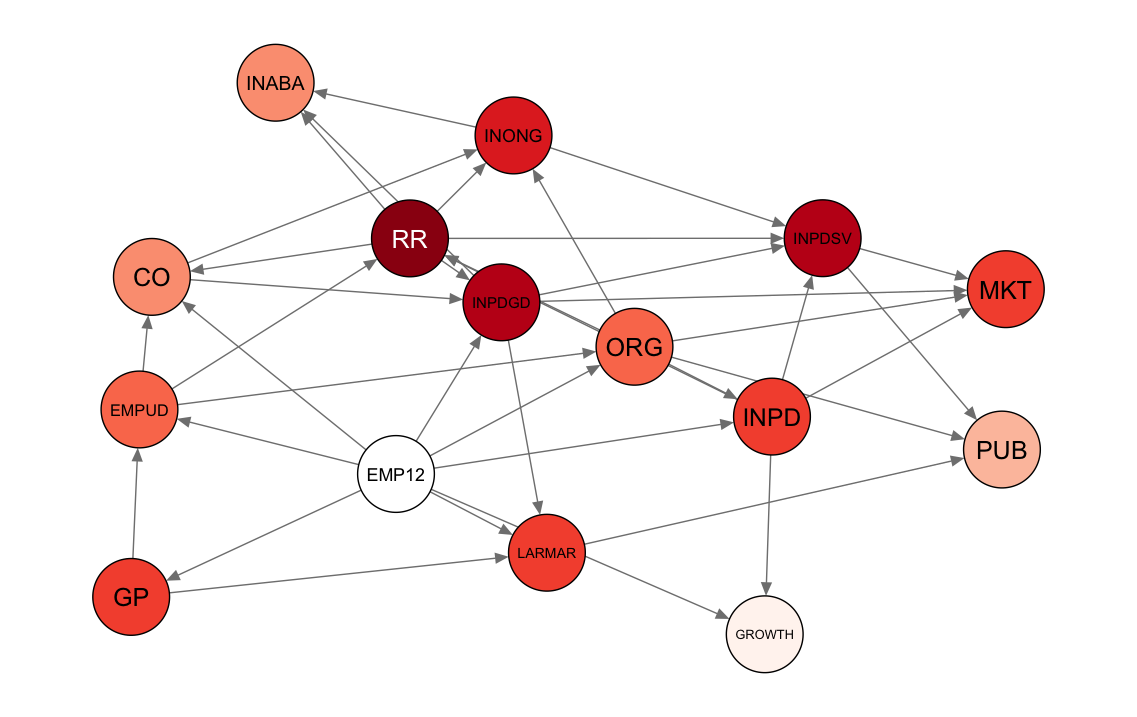}
    \caption{Heatmap of the diameters of the \texttt{istat} BN. Darker colors represent higher diameter values.}
    \label{fig:istat_diameter}
\end{figure}

Tables \ref{tab:asia_diameter} and \ref{tab:istat_diameter} report the diameters of the non-root nodes of the \texttt{asia} and \texttt{istat} BNs. For the \texttt{asia} BN it can be noticed highly different values of the diameter, with the variables \texttt{tub} and \texttt{lung} almost independent of their parents, while the diameter of \texttt{either} is one since it is a deterministic function of its parents (or gate). For the \texttt{istat} BN,  the lowest diameter is the one of the target variable \texttt{GROWTH}, demonstrating that a company's revenue weakly depends on its parents.  The CPT of \texttt{GROWTH} in Table \ref{table:growth} includes very similar rows.  Figure \ref{fig:istat_diameter} gives a visualization of the \texttt{istat} BN diameters, where nodes with a darker color have CPTs whose diameter is larger.

Next we formalize how the diameter behaves under marginalization.

\begin{proposition}
\label{prop:2}
    For three categorical random variables $X,Y,Z$ it holds
    \[
    d^+(P_{Y|X}) \leq d^+(P_{Y|XZ}),
    \]
and equality holds if and only if $Y\independent Z | X$.
\end{proposition}

The intuition behind this result is that each row of $P_{Y|X}$ is a weighted average of some rows of $P_{Y|XZ}$ and therefore the rows of $P_{Y|X}$ are necessarily ``closer" to each other. This implies that the diameter of a CPT based on a subset of the variables in the parent set can always be bounded by the diameter of the CPT based on the complete set of parents.

The last result writes the diameter of a random vector as the sum of the diameter of individual CPTs.
\begin{proposition}
\label{prop:3}
For three categorical random variables $X,Y,Z$ it holds
    \[
    d^+(P_{YZ|X})\leq \min \{d^+(P_{Y|XZ})+d^+(P_{Z|X}),1\}.
    \]
\end{proposition}
These results are useful to construct bounds over the CPTs of vertices of a junction tree, by using the diameters of the original CPTs of the BN, thus not having to compute any new information. We showcase their usefulness in Section \ref{sec:error}. 

In the next section we discuss how the diameter can be used for various sensitivity investigations in BNs which are fully defined (all CPTs have been fully learned or elicited). However, we envisage that the diameter can also be directly elicited from experts given a DAG. The sensitivity methods we develop could then drive the complete elicitation of the BN by focusing on the CPTs which have the biggest impact on the output of interest. A more comprehensive discussion of the use of the diameter within the elicitation of a BN is given in Section \ref{sec:discussion}.

\section{Sensitivity analysis using the diameter}
\label{sec:sens}

\subsection{Edge strength}

\begin{figure}
    \centering
    \includegraphics[scale=0.2]{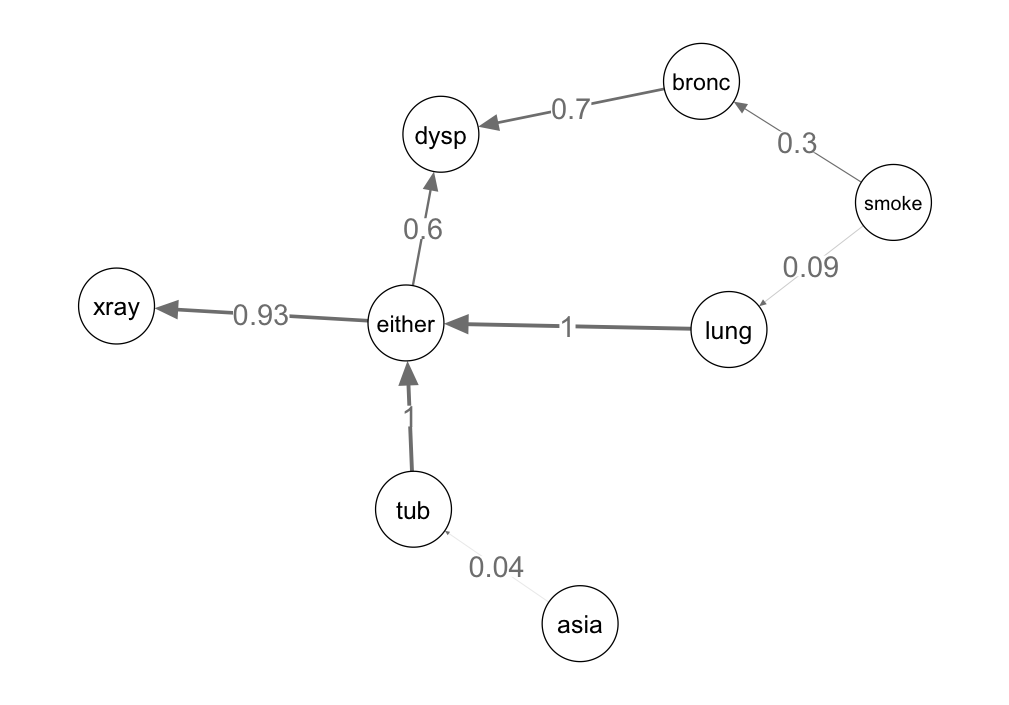}
    \caption{The \texttt{asia} BN with edge strengths as edges' labels and widths.}
    \label{fig:edge_asia}
\end{figure}

The first problem we consider is the quantification of the strength of an edge in a BN. Let $P_i$ be the CPT of $Y_i$ and with $P_{i|\bm{x}}$ we denote the sub-CPT of $P_i$ including only the rows specified by $\bm{x}$. 

\begin{definition}
The strength of edge $(j,i)$ in a BN is defined as
\[
\delta_{ji} = \max_{\bm{x}\in\mathbb{X}_{\Pi_i\setminus j}} d^+(P_{i | \bm{x}}).
\]
\end{definition}
So $\delta_{ji}$ is the largest diameter out of all CPTs for every combination of all parents of $i$ excluding $j$. To illustrate this, consider the CPT of \texttt{GROWTH} in Table \ref{table:growth}. The strength of the edge (\texttt{INPD},\texttt{GROWTH}) is the largest of the diameters of the the three CPTS where \texttt{EMP12} is fixed to its three levels 10-49/50-249/$>250$.

Figures \ref{fig:edge_asia} and \ref{fig:edge_istat} report the edge strengths as edge labels and widths in the \texttt{asia} and \texttt{istat} BN, respectively. They vary from almost zero to exactly one in the case of a functional relationship. Interestingly, the two edges with the lowest strength in the  \texttt{istat} BN are those into the output variable \texttt{GROWTH}, again indicating how the revenue of an enterprise is almost independent of all other factors.

\begin{figure}
    \centering
    \includegraphics[scale=0.25]{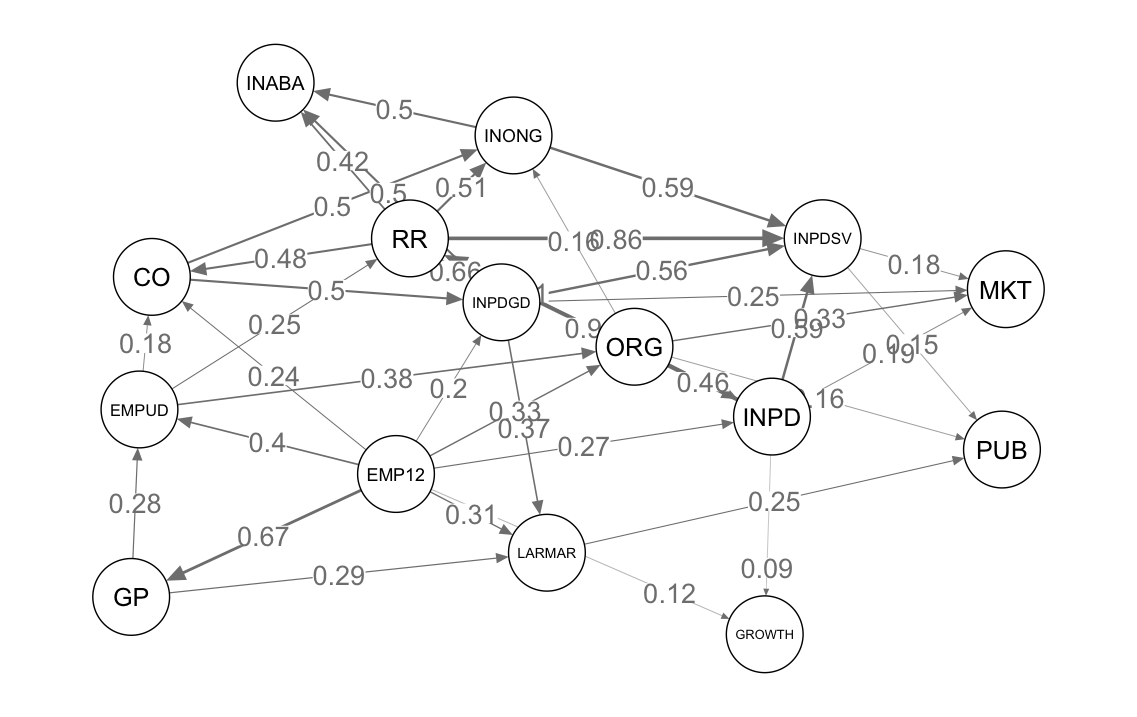}
    \caption{The \texttt{istat} BN with edge strengths as edges' labels and widths.}
    \label{fig:edge_istat}
\end{figure}

Just as the diameter represents marginal independence, edge strength denotes conditional independence.

\begin{proposition}
\label{prop:4}
    \[
    X_i \independent X_j | \bm{X}_{\Pi_i\setminus j} \Leftrightarrow \delta_{ji} = 0.
    \]
\end{proposition}

Thus, in a formal sense, $\delta_{ji}$ is a measure of the extent by which this conditional independence is violated and the merit of knowing the value of $X_j$ to predict $X_i$ once we know the value of $\bm{X}_{\Pi_{i}\setminus j}$.

The following proposition links edge strength to the diameter of a CPT.

\begin{proposition}
\label{prop:5}
    It holds that $\delta_{ji}\leq d^+(P_i)$ and $d^+(P_i)\leq \sum_{j\in\Pi_i}\delta_{ji}$. Also if $|\Pi_i|=1$, then $\delta_{ji}=d^+(P_i)$.
\end{proposition}

Notice that to derive the edge strength $\delta_{ji}$ from elicitation only, $|\mathbb{X}_{\Pi_i\setminus j}|$ diameters must be defined. Edges that appear to be weak could then be discarded before attempting a full quantitative elicitation of the complete CPT, since the size of CPTs increase quadratically with the number of parents.

\subsection{Edge weigthed influence}

Given a DAG and edges labeled with their strength, we may be interested in quantifying the effect of a node on another. For this task, we define a novel measure we henceforth call \textit{edge weighted influence}. Recall that $S_{ji}$ is the set of active, simple trails from $j$ to $i$.

\begin{definition}
    The edge weigthed influence of $X_j$ on $X_i$, $EWI(X_j, X_i)$, is defined as:
    \begin{equation}
        \label{eq:ewi}
        EWI(X_j,X_i)=\sum_{s\in S_{ji}}\left(\prod_{(k,l)\in s} \delta_{kl}\right)^{|s|}.
    \end{equation}
\end{definition}

The definition of the edge weighted influence is inspired by the one of the distance weighted influence where, instead of giving a weight $w$ to every edge of the BN, we consider the edge strengths $\delta_{ji}$. The edge weighted influence sits inbetween mutual information, which requires a full BN definition, and the distance weighted influence, only requiring the DAG, since it is based on the DAG together with some measure of edge strength.

In our examples, the edge weighted distance is derived using the edge strengths computed from the full CPTs, but any measure of edge strength could be considered. For instance, it could be the proportion of times an edge has appeared in non-parametric bootstrap structural learning. An alternative is if these strengths were somehow elicited directly from experts.

\begin{table}[]
    \centering
    \scalebox{0.65}{
    \begin{tabular}{|c|c|c|c|c|c|c|}
    \cline{3-6}
    \multicolumn{1}{c}{}&\multicolumn{1}{c}{} & \multicolumn{4}{|c|}{DWI} &\multicolumn{1}{c}{}\\
    \hline
    Name & Mutual Information & $w=0.1$ & $w=0.2$&  $w=0.5$ & $w=1$ & EWI \\
    \hline
        \texttt{GP} &  0.00132 (6)& 0.01121 (11) & 0.05152 (13)& 0.53125 (14) & 5 (12)  & 0.00645 (4) \\
        \texttt{LARMAR} & 0.00094 (10) & 0.01363 (9) &  0.07683 (9)& 1.70898 (7) & 65 (5) & 0.00145 (8)\\
         \texttt{INPDGD} & 0.00166 (4) & 0.01397
 (7) &0.08270 (7) & 1.91797 (5) & 55 (7) & 0.00075 (10) \\
         \texttt{INPDSV} & 0.00104 (9) & 0.01448 (6) & 0.09372 (5)& 3.05078 (2) & 136 (3) & 0.00351 (6)\\
         \texttt{INPD} & 0.00410 (2) & 0.11111 (1) & 0.24992 (1) & 0.96875 (10) & 5 (12) & 0.09434 (2)\\
         \texttt{INABA}& 0.00029 (13) & 0.00337 (14) & 0.04374 (14)& 2.56641 (4) & 131 (4) & 0.00005 (13)\\
         \texttt{INONG} & 0.00128 (7)  & 0.00508 (12) & 0.05262 (12) & 1.88672 (6)& 57 (6) & 0.00009 (12)\\
         \texttt{CO} & 0.00076 (11) & 0.01390 (8) & 0.08017 (8)& 1.53906 (8) & 31 (8) & 0.00089 (9)\\
         \texttt{ORG} & 0.00147 (5) & 0.02221 (3) & 0.09952 (4) & 0.90625 (12) & 7 (10) & 0.00352 (5) \\
         \texttt{MKT} & 0.00074 (12) & 0.01618 (4) &0.10517 (3) & 3.42188 (1) & 203 (2) & 0.00032 (11)\\
         \texttt{PUB} & 0.00001 (14) & 0.00503 (13) & 0.05401 (11) &2.83301 (3) & 208 (1) & 0.00000 (14)\\
         \texttt{EMP12} & 0.00503 (1) & 0.11111 (1) &  0.24992 (1) & 0.96875 (10) & 5 (12) & 0.12116 (1)\\
         \texttt{EMPUD} & 0.00105 (8) & 0.01321 (10) & 0.06752 (10)& 0.78125 (13) & 7 (10) & 0.00232 (7)\\
         \texttt{RR}  & 0.00277 (3) &  0.01465 (5) & 0.08339 (6)& 1.32812 (9) & 19 (9) & 0.00774 (3) \\
         \hline
        \multicolumn{1}{c}{} &\multicolumn{1}{c|}{}& 0.65566 & 0.52805 & -0.67472 & -0.80443 & 0.81538\\
         \cline{3-7}
    \end{tabular}
    }
\caption{Node influence measures in the \texttt{istat} BN considering \texttt{GROWTH} as output node. In parenthesis the node ranking for each column. The last row reports the Spearman correlation between mutual information and the other influence measures.}
    \label{tab:ewi_istat}
\end{table}

To illustrate the use of the edge weigthed distance we consider the \texttt{istat} BN, since \texttt{asia} has a too simple topology. Table \ref{tab:ewi_istat} reports the mutual information, distance weighted influence for various choices of $w$, and the edge weighted distance between \texttt{GROWTH} and every other node. As already noticed by \citet{Albrecht2014} the $DWI$ measures varies greatly with $w$, with little intuition behind what an optimal $w$ value might actually be. On the other hand, $EWI$ does not require the choice of $w$, of course at the cost of having to specify edge strengths, which on the other hand have a much more straightforward meaning. The highest Spearman correlation with mutual information, reported in the  last row of Table \ref{tab:ewi_istat}, is attained by $EWI$ showcasing how the underlying DAG together with edge strengths gives a very precise approximation to the full BN model.

\begin{figure}
    \centering
    \includegraphics[scale=0.25]{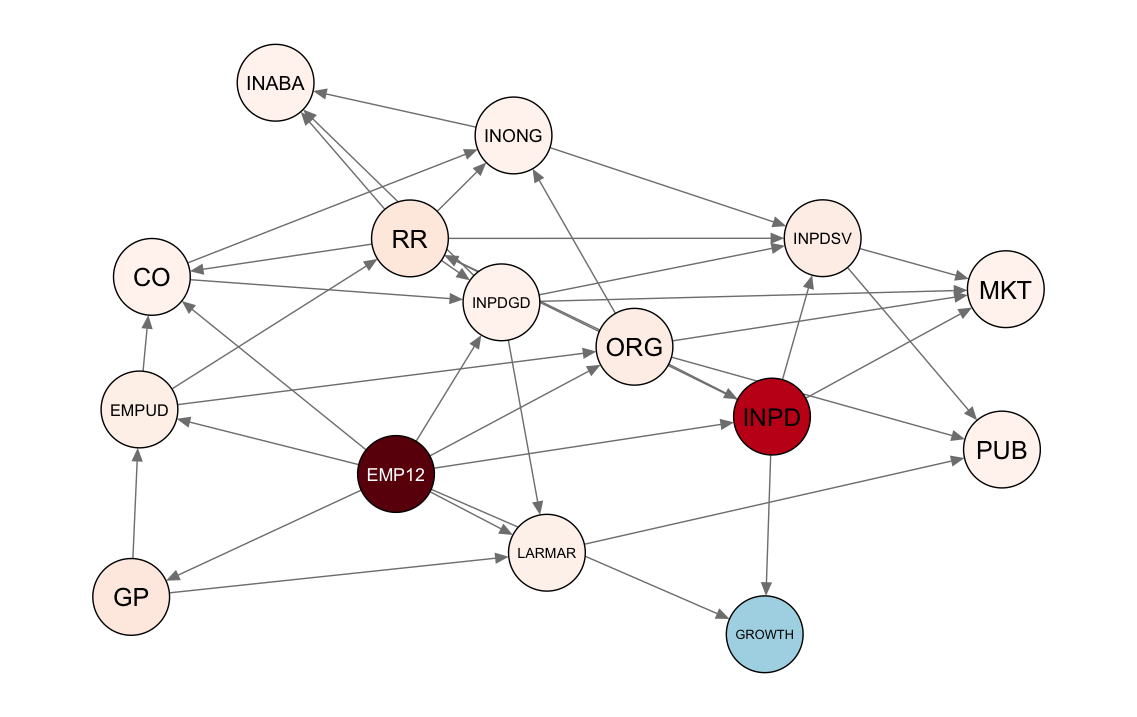}
    \caption{Heatmap of $EWI$ in the \texttt{istat} BN with \texttt{GROWTH} as output (reported in light blue).}
    \label{fig:heatmap_istat}
\end{figure}

The two variables with the greatest effect on \texttt{GROWTH} are, as expected, its parents \texttt{EMP12} and \texttt{INPD}. However, because of the weak dependence between \texttt{GROWTH} and its parents and the actual DAG topology, the EWI of the other nodes quickly vanishes. Figure \ref{fig:heatmap_istat} reports a heatmap of the EWI which clearly provide a visualization of this. Such a heatmap can be highly valuable during elicitation of the entries of the CPTs to convince practitioners on focusing on the specification of the probabilities of the most important factors first.

\subsection{Level amalgamation}

One practical issue found by discrete BN modellers is the number of levels each random variable within the system should be assigned. Obviously there is a trade-off here. The finer the division of levels, the more nuanced the BN can be. On the other hand, the fewer the number of levels, the easier it will be to faithfully elicit or efficiently estimate the probabilities within a BN. The technology of the diameter can be adapted to guide possible amalgamations of the levels of categorical variables, just as when considering whether or not to keep a weak edge in the system.

A first consideration to take into account is that the interpretation of the states can still be understood and quantified by experts. For ordinal variables is therefore recommended to only consider merging consecutive levels (for instance for the variable \texttt{EMP12} with levels 10-49/50-249/$>50$, one should consider merging only the first two or the last two levels). 

The second step in level amalgamation is deciding how to combine the probabilities associated with those levels that are to be amalgamated. Suppose the levels of variable $X_i$ are being merged. In the CPT $P_i$ then the probabilities of the two associated levels are simply summed up. However, also the CPTs of a child of $Y_i$ need to be adapted because of amalgamation. We recommend taking a simple row average, because the convexity of variation distance (Lemma \ref{lemma:3}) guarantees that the diameter of the original CPT does not increase, and, more importantly, this method does not require additional information. In practice we can calculate the diameter of the CPTs with levels amalgamated and combine the closest states first, then find the next closest states and add to the amalgamation iteratively until the combination appears to induce a significant change from the original diameter.

\begin{table}
    \centering\scalebox{0.5}{
    \begin{tabular}{|C{3.5em}|C{6.5em}|C{4em}|C{4em}|C{4em}|C{4em}|C{4em}|C{4em}|C{4em}|C{4em}|}
    \cline{3-10}
    \multicolumn{1}{c}{} &\multicolumn{1}{c|}{} & \texttt{GP}   & \texttt{LARMAR} & \texttt{INPDGD} & \texttt{INPD} & \texttt{CO} & \texttt{ORG} & \texttt{EMPUD} & \texttt{GROWTH} \\
     \hline
  \multirow{3}{*}{\texttt{EMP12}}   & 10-49/50-249    & 0.461 & 0.541 &0.857  & 0.607 & 0.484 & 0.495 & 0.404 & 0.123\\
    & 50-249/$>$250 &  0.538 & 0.597 & 0.788 & 0.577 & 0.376 & 0.456 & 0.500 & 0.153\\
    & Original &0.666 & 0.611 &0.857 & 0.648 & 0.484 & 0.546 & 0.521 & 0.167\\
     \hline
    \end{tabular}}

\vspace{0.3cm}
\scalebox{0.5}{
\begin{tabular}{|C{3.5em}|C{6.5em}|C{4em}|C{4em}|C{4em}|}
    \cline{3-5}
    \multicolumn{1}{c}{} &\multicolumn{1}{c|}{} & \texttt{CO}   & \texttt{ORG} & \texttt{RR} \\
     \cline{1-5}
  \multirow{3}{*}{\texttt{EMPUD}}   & 0\%/$<$10\%    & 0.484 & 0.483 &0.907  \\
    & $<$10\%/$>$10\% &  0.403 & 0.458 & 0.941 \\
    & Original & 0.484 & 0.546 & 0.941\\
     \hline
    \end{tabular}}
  \hspace{1cm}
\scalebox{0.5}{
\begin{tabular}{|C{3.5em}|C{10em}|C{4em}|}
    \cline{3-3}
    \multicolumn{1}{c}{} &\multicolumn{1}{c|}{} & \texttt{PUB}   \\
     \cline{1-3}
  \multirow{3}{*}{\texttt{LARMAR}}   & Regional/National   & 0.111   \\
    & National/International &  0.186  \\
    & Original & 0.344\\
     \hline
    \end{tabular}}
    \caption{Diameters of the CPTs in the \texttt{istat} BN resulting from the amalgamation of levels for the variables \texttt{EMP12}, \texttt{EMPUD}, and \texttt{LARMAR}.}
    \label{tab:amalgamation}
\end{table}

Table \ref{tab:amalgamation} reports the diameters resulting from the amalgamation of levels for variables with more than two levels in the \texttt{istat} BN. It can be seen that merging the levels of \texttt{LARMAR} leads to a strong decrease in the diameter. On the other hand, overall the diameter of the children CPTs of the variables \texttt{EMP12} and \texttt{EMPUD} does not show a strong decrease, highlighting that the amalgamation of a pair of levels would have a small effect on the BN.

\subsection{Asymmetry strength}
\label{sec:asy}
There is now an increasingly amount of evidence that standard conditional independence may be too restrictive to faithfully and fully represent dependence patterns in data \citep{eggeling2019algorithms,leonelli2024structural,pensar2015labeled}. For BNs this means that there are equalities between probability distributions within the CPTs, which therefore have no graphical counterpart in the underlying DAG.

The simplest class of constraints that could appear in a CPT of a BN is the so-called \textit{context-specific} conditional independence \citep{boutilier1996context}. We say that $X_i$ is conditionally independent of $X_j$ given context $X_k=x_k$ if $p(x_i | x_j, x_k) = p(x_i|x_k)$ for all $x_i\in\mathbb{X}_i, x_j\in\mathbb{X}_j$ and a specific $x_k\in\mathbb{X}_k$. A context-specific independence reduces to a standard, symmetric independence if it holds for all $x_k\in\mathbb{X}_k$. Consider the CPT in Table \ref{table:example} of a random variable $X_i$ conditional on $X_j$ and $X_k$, all taking levels high, medium, and low. It can be seen that $X_i$ is conditionally independent of $X_j$ when $X_k = \text{high}$ since all rows where $X_k = \text{high}$ have the same pmf.

\begin{table}[]
    \centering
    \scalebox{0.6}{
    \begin{tabular}{|C{3.5em}|C{3.5em}|C{5.5em}|C{5.5em}|C{5.5em}|}
    \hline
    \multirow{2}{*}{$X_j$} & \multirow{2}{*}{$X_k$} & \multicolumn{3}{c|}{$P(X_i|X_j,X_k)$} \\
    \cline{3-5}
         & & high & medium & low \\
         \hline
    high & high &  0.500 & 0.300 & 0.200\\
    \hline
    high & medium & 0.400 & 0.200 & 0.500\\
    \hline 
    high & low  &  0.300 & 0.100 & 0.600\\
    \hline 
    medium & high & 0.500 & 0.300 & 0.200 \\
    \hline
    medium & medium  &  0.300 & 0.500 & 0.200\\
    \hline 
    medium & low &  0.200 & 0.400 & 0.400 \\
    \hline 
    low & high & 0.500 & 0.300 & 0.200\\
    \hline
    low & medium & 0.200 & 0.200 & 0.600\\
    \hline
    low & low & 0.200 & 0.200 & 0.600 \\
    \hline
    \end{tabular}
    }
    \caption{Example of a CPT embedding context-specific and partial conditional indepenences.}
    \label{table:example}
\end{table}

\citet{pensar2016role} introduced a more general extension to conditional independence called \textit{partial conditional independence}. We say that $X_i$ is partially conditionally independent of $X_j$ in the domain $\mathcal{D}_j\subseteq \mathbb{X}_j$ given context $X_k=x_k$ if $p(x_i | x_j, x_k)=p(x_i|\tilde{x}_j,x_k)$ holds for all $(x_i,x_j),(x_i,\tilde{x}_j)\in\mathbb{X}_i\times\mathcal{D}_j$. Partial and context-specific independence coincides if $\mathcal{D}_j=\mathbb{X}_j$. Furthermore, the sample space $\mathbb{X}_j$ must contain more than two elements for a non-trivial partial conditional independence to hold. The CPT in Table \ref{table:example} embeds the partial conditional independence between $X_i$ and $X_k$ in the domain $\{\text{medium},\text{low}\}$ in the context $X_j=\text{low}$.

Currently, there are no methods to assess whether additional equalities are present in the CPTs of a BN, beyond visual investigation. Here we demonstrate that the diameter can be also be used for this task,  just as when considering whether or not
to keep a weak edge in the system.  

\begin{definition}
The index of context-specific independence between $X_i$ and $X_j$, for $j\in \Pi_i$, in the context $\bm{x}\in\mathbb{X}_{\Pi_i\setminus j }$ is
\[
\delta_{\bm{x}i}^+ = \ d^+(P_{i | \bm{x}}),
\]
and the index of partial independence is
\[
\delta_{\bm{x}i}^- = \ d^-(P_{i | \bm{x}}).
\]
\end{definition}

The index of context-specific independence computes the upper diameter of a CPT where all parents but one are fixed, while the index of partial independence computes the lower diameter for the same CPT. If the lower diameter is close to zero it means that there are at least two rows of the CPT which are very similar to each other, thus implying a partial conditional independence.

\begin{table}[]
    \centering
    \scalebox{0.6}{
    \begin{tabular}{|c|c|}
    \hline
    \multicolumn{2}{|c|}{\texttt{GP}}\\
    \hline
    \texttt{EMP12} & Context\\
    \hline
      10-49   &  0.255\\
       50-249  &   0.224 \\
       $>250$ &  0.284 \\
       \hline
    \end{tabular}}
    \hspace{0.5cm}
       \scalebox{0.6}{
        \begin{tabular}{|c|c|c|}
    \hline
    \multicolumn{3}{|c|}{\texttt{EMP12}}\\
    \hline
    \texttt{GP} &Context & Partial\\
    \hline
      No   & 0.398& 0.007\\
       Yes  &   0.265 &0.131 \\
       \hline
    \end{tabular}}
    \caption{Context-specific and partial independence indices for the CPT of the variable \texttt{EMPUD} in the \texttt{istat} BN.}
    \label{table:asy}
\end{table}

Propositions \ref{prop:4} and \ref{prop:5} could be straightforwardly extended to link the context-specific index to context-specific independence and bound the index with the diameter of the CPT.

As an illustration Table \ref{table:asy} reports the context-specific and partial indices for the CPT of the variable \texttt{EMPUD} in the \texttt{istat} BN. The context-specific independence indices are quite close to each other, with a very strong dependence between \texttt{EMP12} and \texttt{EMPUD} in the context \texttt{GP} = No. The partial independence index is only reported for the \texttt{EMP12} variables since it is ternary. It can be seen that in the context \texttt{GP} = No it is quite plausible the presence of a partial independence: the probability distribution of \texttt{EMPUD} is very similar for at least two levels of \texttt{EMP12} in the context \texttt{GP} = No. \ref{sec:staged} shows that a model selection algorithm for a class of models that formally embeds also non-symmetric types of independence does indeed learn the mentioned partial independence. 

\section{Error propagation}
\label{sec:error}

It is well known that when using standard propagation algorithms on updating one of the clique margins, say $C_i$, the knock-on effect on the other clique margins becomes weaker and weaker as the updated cliques become progressively more remote from $C_i$. This property is exploited by \citet{Albrecht2014} when defining the DWI influence, for instance, and thus also in the EWI distance. It is also known that if the underlying DAG is a polytree, the mutual influence between two nodes decreases with the distance (number of edges) between them.

The extent of the deviation can be bounded using variation distance, providing an upper limit to the potential error in the distributions of focus variables induced by the misspecification of various CPTs in the BN. This is particularly useful when we elicit a large BN and want to know how far away from target nodes we need to elicit the corresponding CPTs accurately. The following result provides the basis for the derivation of such bounds.

\begin{theorem}
\label{theo:1}
Consider two categorical random variables $X$  and $Y$. Let $p$ and $p'$ two pmfs over $(X,Y)$ such that $p(y|x)=p'(y|x)$ for all $x\in\mathbb{X}$ and $y\in \mathbb{Y}$. Then
\[
d_V(p(y),p'(y))\leq d^+(P_{Y|X})d_V(p(x),p'(x)).
\]
\end{theorem}

The interpretation of this result is as follows. Suppose the CPT of $Y$ given $X$ has been specified accurately, but that the margin probability of $Y$ is uncertain. Then our marginal beliefs about $Y$ are no more uncertain than those about $X$, because by definition $d^+(P_{Y|X})\leq 1$. More importantly, we have a bound on how much the uncertainty, quantified in terms of total variation, reduces in terms of $d^+(P_{Y|X})$ - a measure of how far away $Y$ is from independence of $X$.

We can now use Theorem \ref{theo:1} to provide a bound of the effect of perturbation of an output variable on the cliques of the junction tree of a BN.

\begin{theorem}
    Consider an output variable $X_j$ which has no children in the BN and assume it is in clique $C_j$ with no loss of generality. Let $C_i$ be another clique and $S_{i+1},\dots,S_{j}$ be the separators along the unique path between $C_i$ and $C_j$. Let $S_k^*=S_k\setminus \cup_{l=k+1}^{j}S_l $, for $k=i+1,\dots,j$, $P_k^*$ be the CPT representing $p(\bm{x}_{S^*_k}|\bm{x}_{S^*_{k-1}})$ for $k=i+2,\dots,j$ and $P_j$ be the CPT of $p(\bm{x}_{j}|\bm{x}_{S^*_{j}})$. Then for any two pmfs $p$ and $p'$ we have that
    \[
    d_V(p(x_j),p'(x_j))\leq d^+(P_j)\prod_{k=i+2}^j d^+\left(P^*_k\right)d_V(p(\bm{x}_{C_i}),p'(\bm{x
}_{C_i})).
    \]
\end{theorem}

The proof of this result follows easily from Equation (\ref{eq:markov}) and successive application of Theorem \ref{theo:1}. Notice that the quantities $d^+(P^*_k)$ may or may not be readily available from a fully defined BN. In the best case scenario, their computation may simply require the use of Propositions \ref{prop:2} and \ref{prop:3} which take advantage of simple properties of the diameter. In other cases, when a separator $S_{k+1}$ includes parents of variables in $S_{k}$, the required CPT must be constructed using inferential queries over the BN.

Next we define the impact of one clique upon a variable of interest in order to ascertain the diminishing effect of errors downstream in the chain of a junction tree.

\begin{definition}
    The impact of $C_i$ on an output variable $Y_j$ in clique $C_j$ is $d^+(P_j)\prod_{k=i+2}^j d^+\left(P^*_k\right)$.
\end{definition}

The impact of one clique on an output is a simple measure of the maximum possible influence the misspecification of one set of probabilities could have on another as measured by a bound on the variation distance. As an illustration, Figure \ref{fig:junction_color} shows the junction tree of the \texttt{asia} BN together with the clique impact when the variable \texttt{xray} is the output of interest. The impact of the clique \texttt{either}, \texttt{lung}, \texttt{tub} is 0.93, which is simply the diameter of the CPT of \texttt{xray}. The impact of the cliques \texttt{asia}, \texttt{tub} and \texttt{lung}, \texttt{either}, \texttt{bronc} is also 0.93 since this is equal to the diameter of the CPT of \texttt{xray} times the diameter of the CPT of \texttt{either}, which is equal to 1. The remaining cliques have impact $0.93\cdot 1 \cdot 0.09$: that is, the impact of \texttt{bronc}, \texttt{lung}, \texttt{either} times the diameter of the CPT of \texttt{lung}. 

\begin{figure}
    \centering
    \includegraphics[scale=0.25]{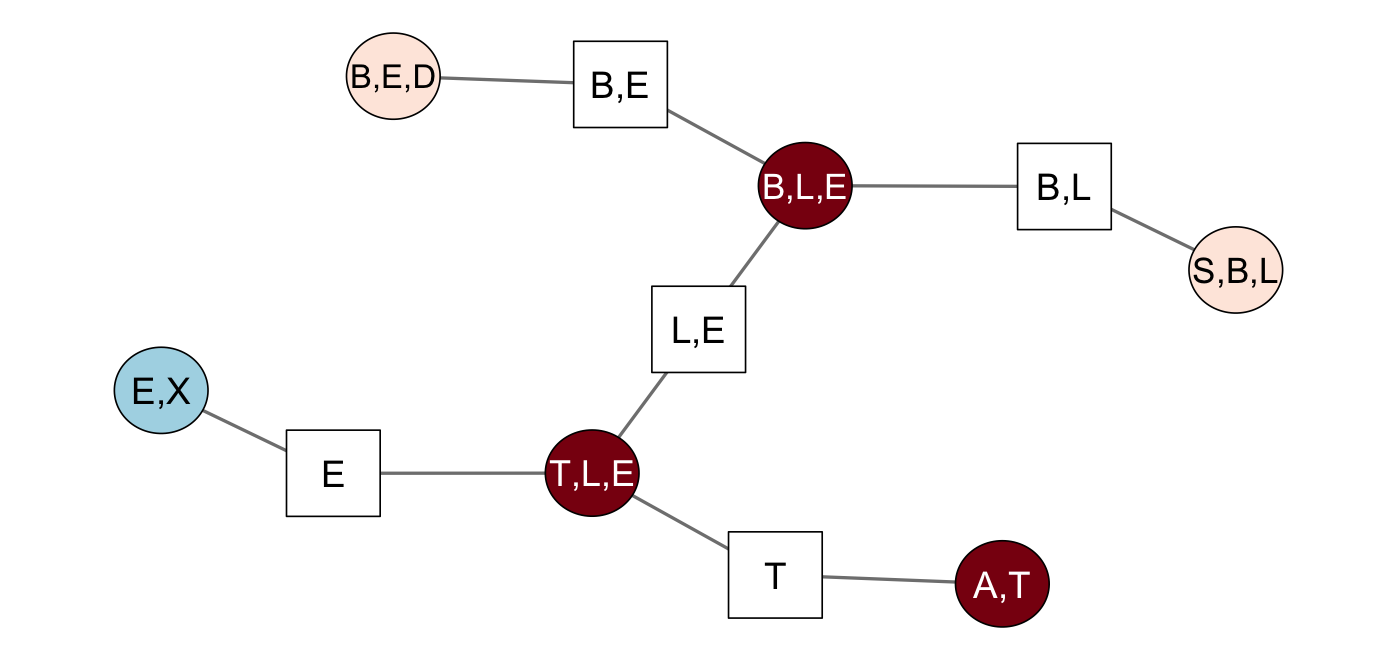}
    \caption{The junction tree representing the \texttt{asia} BN where cliques are colored by their impact on \texttt{xray}.}
    \label{fig:junction_color}
\end{figure}

Notice that because the diameter is bounded by one, we have the following nice property, confirming that cliques further away from the output have a smaller effect. 

\begin{proposition}
    The impact of a clique $C_i$ on an output variable $Y_j$ in clique $C_j$ is smaller than the impact of $C_k$ on $C_j$ for any $C_k$ along the unique path between $C_i$ and $C_j$.
\end{proposition}

These influences provide a very useful tool for prioritisation of the elicitation in a BN. If we can obtain estimates of influence across a junction tree (either from direct elicitation or alternatively after having performed a preliminary coarse elicitation of the corresponding CPTs) then we can use these influences to identify which of those CPTs to refine. A further discussion of how the results of this paper can be used for elicitation is next.

\section{The diameter as an elicitation tool}

We have hinted throughout this paper at possible ways in which the diameter can be used to evaluate the state of a BN elicitation from experts. These evaluations can be embedded into a formal protocol. However, there are many considerations that a user needs to make before undertaking model construction: transparency of the model, computational issues, elicitation constraints and so forth, which vary in importance depending on the context of the model. So, setting a bound on any effects or perturbations against differing approaches is often best undertaken informally. We acknowledge that the framework we have presented here is sufficiently formal to admit generalisation, and this is work that we plan to undertake in the future.

To implement our techniques as efficiently as possible, we recommend two differing approaches tailored to the specific circumstances of the modeller. Firstly, there are occasions in which we have obtained provisional information from one expert who can recommend nodes, levels, interactions and provisional CPTs before undertaking a more formal elicitation conference with multiple experts \citep[see e.g.][]{barons2018assessment}. In this particular scenario, we can begin to design the analysis by using the bounds discussed earlier on the preliminary values stated by the expert. We recommend starting by eliciting attributes and nodes of interest before working systematically backwards along the chain of inference to discover parent nodes and conditional independences, performing variation measures on preliminary CPT values to determine the efficacy of including variables in the model. Of course, after the complete elicitation has taken place, the robustness analyses suggested above can be repeated for a final sensitivity analysis.

In situations where we are starting the model with no such preliminary information, it may be wisest to attempt to elicit the value of the diameter of the CPT directly before eliciting the entire matrix so that complete elicitation is not undertaken before we can derive concrete bounds on the usefulness of this data harvesting exercise. To elicit the diameter directly, we need to ascertain the largest differences between rows of a CPT, which corresponds to requesting the “best case scenario” and “worst case scenario” probabilities before calculating the variation distance between the two.

By following this procedure, we, therefore, continuously appraise and compare each possible simplification against the potential accuracy of an analysis, weighted against the issues provided by a simpler model representation. We have demonstrated above that in many cases, the effects of various simplifications are often minimal, and approximations based on these simplifications are justified from a pragmatic point of view. We also note that some of the best approximations frequently differ from those currently undertaken in practice. For example, using an approximation that deletes an edge can cause significant changes whilst allowing dependence only on subsets of levels, which performs much better.

\section{Discussion}
\label{sec:discussion}

We have demonstrated here how the properties of variation distance can be harnessed to study the robustness of a discrete BN, if certain target variables are known a priori to be those of primary interest. We took advantage of the fact that the variation distance naturally embeds conditional independence relationships between variables. This observation enabled us to devise a seamless way of looking at perturbed versions of a BN for various sensitivity problems, including edge strength, node influence, level amalgamation and asymmetric dependence. An implementation of the developed methods is freely available to practitioners through the \texttt{bnmonitor} R package.

Much work is still to be undertaken, starting with refining the bounds we have developed here. Similarly, within this paper, we have had no space to consider the robustness of the choice of probability distribution on the entries of the CPTs of a BN. In \citet{smith2010robustness}, BN robustness associated with the inputs of the distribution in terms of the local DeRobertis distance \citep{deroberts1981bayesian} is studied, while \citet{smith2012isoseparation} provided bounds on posterior variation distances. Therefore, a reasonably straightforward extension to the variation bounds we have presented here can be developed by carefully combining our results with the local DeRobertis distance to provide a comprehensive robustness analysis. 

Our ideas also apply directly to dynamic BNs where the system's sensitivity can be even more critical because the dynamic nature of the problem makes the model much more complex. Throughout this paper, for simplicity, we have considered only robustness as it applies to finite categorical BNs. However, using the approach developed here, the variation distance on these highly structured and complex Markov processes can help us determine the robustness of dynamic BNs to dynamic effects. Furthermore, the diameter could be extended to perform sensitivity studies in mixed BNs, where both categorical and continuous variables are considered.

\bibliographystyle{elsarticle-harv} 
 \bibliography{cas-refs}






\appendix

\section{Auxiliary results}

In this appendix we present results that will be used in the proofs of the main propositions and theorems of the paper.

\begin{lemma}
\label{lemma:1}
Let $p$ and $p'$ be two pmfs over the same sample space $\mathbb{X}$. Then
\begin{equation}
\label{eq:lemma1}
\sum_{x\in\mathbb{X}}\min\{p(x),p'(x)\}= 1- d_V(p,p')
\end{equation}
\end{lemma}
\proof{}
By definition of total variation distance we have
\begin{multline*}
\sum_{x\in\mathbb{X}}\min\{p(x),p'(x)\} + d_V(p,p')  =  \sum_{x\in\mathbb{X}}\min\{p(x),p'(x)\} \\+ \frac{1}{2}\sum_{x\in\mathbb{X}}|p(x)-p'(x)|.
\end{multline*}
Let $\mathbb{X}_1$ be the subset of $\mathbb{X}$ including those $x$ such that $p(x)\leq p'(x)$ and $\mathbb{X}_2=\mathbb{X}\setminus \mathbb{X}_1$.
Therefore
\begin{align*}
    \sum_{x\in\mathbb{X}}\min\{p(x),p'(x)\} + d_V(p,p')  = &\sum_{x\in\mathbb{X}_1}p(x)+\sum_{x\in\mathbb{X}_2}p'(x)+\\
    & \frac{1}{2}\sum_{x\in\mathbb{X}_1}(p'(x)-p(x))+\frac{1}{2}\sum_{x\in\mathbb{X}_2}(p(x)-p'(x))\\
    = & \frac{1}{2}\sum_{x\in\mathbb{X}_1}p(x)+\frac{1}{2}\sum_{x\in\mathbb{X}_1}p'(x)+\\
    &\frac{1}{2}\sum_{x\in\mathbb{X}_2}p'(x)+\frac{1}{2}\sum_{x\in\mathbb{X}_2}p(x)\\
    = & \frac{1}{2}+\frac{1}{2}=1
\end{align*}
\endproof

\begin{lemma}
Let $p$ and $p'$ be two pmfs over the same sample space $\mathbb{X}$. Call
\begin{align*}
    \beta & = 1- d_V(p,p')\\
    p^* &=\frac{\min\{p,p'\}}{\beta}\\
    \bar{p}&=\frac{p-\min\{p,p'\}}{1-\beta}\\
    \bar{p}'&=\frac{p'-\min\{p,p'\}}{1-\beta}
\end{align*}
Then 
\begin{equation}
    \label{eq:lemma2}
p=\beta p^* + (1-\beta)\bar{p}
\end{equation}
and similarly for $p'$
\label{lemma:2}
\end{lemma}

\proof{}
First notice that from Lemma \ref{lemma:1}, both $p^*$ and $\bar{p}$ are pmfs. For $p^*$ this follows directly from (\ref{eq:lemma1}). For $\bar{p}$ notice that
\begin{align*}
\sum_{x\in\mathbb{X}}\frac{p(x)-\min\{p(x),p'(x)\}}{1-\beta} &=\frac{1}{d_V(p,p')}\sum_{x\in\mathbb{X}}\left(p(x)-\min\{p(x),p'(x)\}\right)\\
&=\frac{1}{d_V(p,p')}(1-(1-d_V(p,p')))=1.
\end{align*}
By substituting the definitions of $p^*$ and $\bar{p}$ into (\ref{eq:lemma2}) the result follows.
\endproof

\begin{lemma}
\label{lemma:3}
Let $p,q_1,\dots,q_{|\mathbb{X}|}$ be pmfs over the same sample space $\mathbb{X}$. Let $\pi=(\pi_1,\dots,\pi_{|\mathbb{X}|})\in\Delta_{|\mathbb{X}|-1}$ and $q_\pi =\sum_{i=1}^{|\mathbb{X}|}\pi_iq_i$. Then
\[
d_V(p,q_\pi)\leq \sum_{i=1}^{|\mathbb{X}|}\pi_id_V(p,q_i)
\]
\end{lemma}

\begin{proof}
\begin{align*}
    d_V(p,q_\pi)&=\frac{1}{2}\sum_{x\in\mathbb{X}}|p(x)-q_\pi(x)|= \frac{1}{2}\sum_{x\in\mathbb{X}}\left|p(x)-\sum_{i=1}^{|\mathbb{X}|}\pi_iq_i(x)\right|\\
    &= \frac{1}{2}\sum_{x\in\mathbb{X}}\left|\sum_{i=1}^{|\mathbb{X}|}\pi_ip(x)-\sum_{i=1}^{|\mathbb{X}|}\pi_iq_i(x)\right|= \frac{1}{2}\sum_{x\in\mathbb{X}}\left|\sum_{i=1}^{|\mathbb{X}|}\pi_i(p(x)-q_i(x))\right|\\
    &\leq \frac{1}{2}\sum_{x\in\mathbb{X}}\sum_{i=1}^{|\mathbb{X}|}\pi_i|p(x)-q_i(x)|=\sum_{i=1}^{|\mathbb{X}|}\pi_i \frac{1}{2}\sum_{x\in\mathbb{X}}|p(x)-q_i(x)|\\
    &=\sum_{i=1}^{|\mathbb{X}|}\pi_id_V(p,q_i)
\end{align*}
\end{proof}

For the next lemma we consider two categorical random variables $X$ and $Y$  with sample spaces $\mathbb{X}$, $\mathbb{Y}$, respectively. With $p(X)$, $p(X,Y)$ and $p(Y|x)$  we denote the marginal, joint and conditional probability distributions over the corresponding variables. The same notation is used for another pmf $p'$.

\begin{lemma}
It holds that
    \label{lemma:0}
    \[
    d_V(p(X,Y),p'(X,Y))\leq \min\{d_V(p(X),p'(X)) + \max_{x\in\mathbb{X}}d_V(p(Y|x),p'(Y|x)),1\}
    \]
\end{lemma}
\begin{proof}
First notice that
\begin{align*}
      d_V(p(X,Y),p'(X,Y))&=\frac{1}{2}\sum_{x\in\mathbb{X},y\in\mathbb{Y}}|p(x,y)-p'(x,y)|\\
      &=\frac{1}{2}\sum_{x\in\mathbb{X},y\in\mathbb{Y}}|p(y|x)p(x)-p'(y|x)p(x)|.
\end{align*}
Calling $r(y|x)=p(y|x)-p'(y|x)$, we have
\begin{align*}
    d_V(p(X,Y),p'(X,Y))=&\hspace{0.3cm}\frac{1}{2}\sum_{x\in\mathbb{X},y\in\mathbb{Y}}|p(y|x)(p(x)-p'(x))+r(y|x)p'(x)|\\
    \leq&\hspace{0.3cm}\frac{1}{2}\sum_{\substack{x\in\mathbb{X}\\y\in\mathbb{Y}}}|p(y|x)(p(x)-p'(x))|+\frac{1}{2}\sum_{\substack{x\in\mathbb{X}\\y\in\mathbb{Y}}}|r(y|x)p'(x)|\\
    =&\hspace{0.3cm}\frac{1}{2}\sum_{x\in\mathbb{X}}|p(x)-p'(x)|\sum_{y\in\mathbb{Y}}p(y|x) + \\
     &\hspace{0.3cm}\frac{1}{2}\sum_{x\in\mathbb{X}}p'(x)\sum_{y\in\mathbb{Y}}|r(y|x)|\\
     \leq&\hspace{0.3cm} \frac{1}{2}\sum_{x\in\mathbb{X}}|p(x)-p'(x)| + \frac{1}{2}\max_{x\in\mathbb{X}}\sum_{y\in\mathbb{Y}}|r(y|x)|\\
     =&\hspace{0.3cm}d_V(p(X),p'(X))+\max_{x\in\mathbb{X}}d_V(p(Y|x),p'(Y|x)).
\end{align*}
Since the maximum value for the total variation distance is one, the result follows.
\end{proof}

\section{Proofs}
\subsection{Proof of Proposition \ref{prop:1}}
For the $(\Rightarrow)$ statement, we notice that if $d^+(P_{Y|X})=0$ then any two rows of $P_{Y|X}$ must be the same. Thus $p(Y|x)=p(Y)$ for all $x\in\mathbb{X}$. Thus $Y\independent X$. For the $(\Leftarrow)$ statement, if $Y\independent X$ then $p(Y|x)=p(Y)$ for all $x\in\mathbb{X}$. Thus every row of $P_{Y|X}$ is equal to $p(Y)$ and its diameter is zero.

 \subsection{Proof of Proposition \ref{prop:2}}

Consider any two $x,x'\in\mathbb{X}$. Using twice Lemma \ref{lemma:3} we have that 
\begin{align*}
    d_V(p(Y|x),p(Y|x')) & =d_V\left(\sum_{z\in\mathbb{Z}}p(Y|x,z)p(z|x),\sum_{z'\in\mathbb{Z}}p(Y|x',z')p(z'|x')\right) \\
   & \leq \sum_{z'\in\mathbb{Z}}p(z'|x')d_V\left(\sum_{z\in\mathbb{Z}}p(Y|x,z)p(z|x),p(Y|x',z')\right)\\
   &\leq \sum_{z'\in\mathbb{Z}}p(z'|x')\sum_{z\in\mathbb{Z}}p(z|x)d_V\left(p(Y|x,z),p(Y|x',z')\right)\\
   &\leq d^+(P_{Y|(x,x')Z})\sum_{z'\in\mathbb{Z}}p(z'|x')\sum_{z\in\mathbb{Z}}p(z|x)\\
   & = d^+(P_{Y|(x,x')Z}),
\end{align*}
where $P_{Y|(x,x')Z}$ is the CPT of $Y$ conditional on all $z\in\mathbb{Z}$ but only on $x,x'\in\mathbb{X}$. The result follows since  \[
d^+(P_{Y|X})=\max_{x,x'\in\mathbb{X}}d_V(p(Y|x),p(Y|x'))\leq \max_{x,x'\in\mathbb{X}} d^+(P_{Y|(x,x')Z})\leq d^+(P_{Y|XZ}).
\]
The proof about the equality between diameters follows a similar reasoning to that of Proposition \ref{prop:1}.

\subsection{Proof of Proposition \ref{prop:3}}
For any two $x,x'\in\mathbb{X}$, Lemma \ref{lemma:0} can be rewritten as
\begin{multline*}
    d_V(p(Y,Z|x), p'(Y,Z|x'))  \leq d_V(p(Z|x),p'(Z|x')) + \\\max_{z\in\mathbb{Z}} d_V(p(Y|z,x), p'(Y|z,x')).
\end{multline*}
Notice that $\max_{z\in\mathbb{Z}} d_V(p(Y|z,x), p'(Y|z,x'))\leq d^+(P_{Y|ZX})$ and therefore
\begin{multline*}
    d^+(P_{YZ|X})=\max_{x,x'\in\mathbb{X}}  d_V(p(Y,Z|x), p'(Y,Z|x')) \\ \leq d^+(P_{Y|ZX}) + \max_{x,x'\in\mathbb{X}} d_V(p(Z|x),p'(Z|x')).
\end{multline*}
Noticing that $\max_{x,x'\in\mathbb{X}} d_V(p(Z|x),p'(Z|x'))=d^+(P_{Z|X})$, the result follows.
 
\subsection{Proof of Theorem \ref{theo:1}}

Consider the total variation distance between $p(Y)$ and $p'(Y)$. Using the definition of conditional probability this is equal to
\[
d_V(p(Y),p'(Y))=d_V(p(X)P_{Y|X},p'(X)P_{Y|X})
\]
Let $P_{Y|X}(x)$ be the $x$-th row of $P_{Y|X}$ and notice that $p(X)P_{Y|X}=\sum_{x\in\mathbb{X}}p(x)P_{Y|X}(x)$.
Using Lemma \ref{lemma:2} we now derive that
\begin{align*}
    d_V(p(Y),p'(Y)) & = 
    d_V((\beta p^* + (1-\beta)\bar{p})P_{Y|X},(\beta p^* + (1-\beta)\bar{p}')P_{Y|X})\\
    & = \frac{1}{2}\sum_{x\in\mathbb{X}}|(1-\beta)(\bar{p}(x)-\bar{p}'(x))P_{Y|X}(x)|\\
    & = \frac{1}{2}\sum_{x\in\mathbb{X}}|1-\beta||(\bar{p}(x)-\bar{p}'(x))P_{Y|X}(x)|\\
    & = |1-\beta|d_V(\bar{p}P_{Y|X},\bar{p}'P_{Y|X})\\
    &=(1-\beta)d_V(\bar{p}P_{Y|X},\bar{p}'P_{Y|X}),
\end{align*}
where the last equality follows by the fact that the total variation distance is at most 1.

Using again the fact that $p(X)P_{Y|X}=\sum_{x\in\mathbb{X}}p(x)P_{Y|X}(x)$ together with Lemma \ref{lemma:3}, we get that
\begin{align*}
    d_V(p(Y),p'(Y)) & = (1-\beta)d_V(\bar{p}P_{Y|X},\bar{p}'P_{Y|X})\\
    &=(1-\beta)d_V\left(\sum_{x\in\mathbb{X}}\bar{p}(x)P_{Y|X}(x),\sum_{x'\in\mathbb{X}}\bar{p}'(x')P_{Y|X}(x')\right)\\
    &\leq (1-\beta)\sum_{x'\in\mathbb{X}}\bar{p}'(x')d_V\left(\sum_{x\in\mathbb{X}}\bar{p}(x)P_{Y|X}(x),P_{Y|X}(x')\right)\\
    &\leq (1-\beta)\sum_{x'\in\mathbb{X}}\sum_{x\in\mathbb{X}}\bar{p}'(x')\bar{p}(x)d_V\left(P_{Y|X}(x'),P_{Y|X}(x)\right)\\
    &\leq (1-\beta)\sum_{x'\in\mathbb{X}}\sum_{x\in\mathbb{X}}\bar{p}'(x')\bar{p}(x)d^+_{Y|X},
\end{align*}
where in the last inequality we used the definition of diameter of a CPT. Using the fact that $\sum_{x'\in\mathbb{X}}\sum_{x\in\mathbb{X}}\bar{p}'(x')\bar{p}(x)=1$ and that $1-\beta=d_V(p(X),p'(X))$, the result follows.

\section{Refining a Bayesian networks with asymmetric independence}

In this appendix we showcase the use of a context-specific class of models, namely staged trees \citep{collazo2018chain,smith2008conditional}, to learn conditional independence relationships from data that are not simply symmetric. Here we simply give an illustration of how a staged tree would represent the structure of the CPT of the variable \texttt{EMPUD} in the \texttt{istat} BN. With this aim we construct an event tree with the parents of \texttt{EMPUD}, namely \texttt{GP} and \texttt{EMP12}, and \texttt{EMPUD} itself, reported in Figure \ref{fig:enter-label}.

Each vertex at depth two in this tree represents the conditional distribution of \texttt{EMPUD} for a specific combination of its parents: it thus represents a row of the CPT of \texttt{EMPUD}. A staged tree learning algorithm then aims at finding the best partition of these vertices, where two vertices are in the same set if their pmfs are assumed to be equal. When this happens two vertices are said to be in the same \textit{stage} and are equally colored in the tree graph. We used the \texttt{stagedtrees} R package \citep{carli2022r} to learn an optimal partition of the vertices associated to the pmf of \texttt{EMPUD}, reported in the coloring of the tree in Figure \ref{fig:enter-label}. It can be noticed that the only two vertices that are given the same color are for the context \texttt{GP} = No and two levels of \texttt{EMP12}. This correspond to the setup where the index of partial independence was really low in Table \ref{table:asy}. The fact that no other pair of vertices have the same color suggests that there are no other equalities between the rows of the CPTs of \texttt{EMPUD}, confirming what suggested by the indices in Table \ref{table:asy}.

\label{sec:staged}
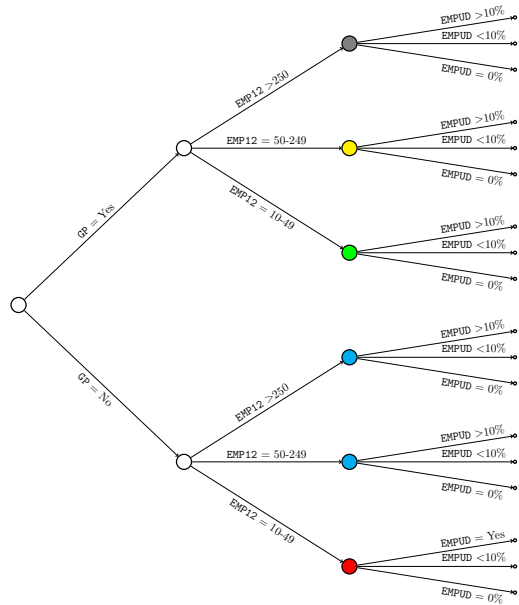
\begin{figure}
    \centering
    \scalebox{0.4}{
    \begin{tikzpicture}[auto, scale=10,
	ORG_NA/.style={circle,inner sep=1mm,minimum size=0.5cm,draw,black,very thick,fill=white,text=black},
	EMPUD_1/.style={circle,inner sep=1mm,minimum size=0.5cm,draw,black,very thick,fill={rgb,255:red,0; green,0; blue,0},text=black},
	EMPUD_2/.style={circle,inner sep=1mm,minimum size=0.5cm,draw,black,very thick,fill=cyan,text=black},
	INPD_1/.style={circle,inner sep=1mm,minimum size=0.5cm,draw,black,very thick,fill=red,text=black},
	INPD_3/.style={circle,inner sep=1mm,minimum size=0.5cm,draw,black,very thick,fill=green,text=black},
	INPD_4/.style={circle,inner sep=1mm,minimum size=0.5cm,draw,black,very thick,fill=yellow,text=black},
	INPD_6/.style={circle,inner sep=1mm,minimum size=0.5cm,draw,black,very thick,fill=gray,text=black},
	RR_1/.style={circle,inner sep=1mm,minimum size=0.5cm,draw,black,very thick,fill=red,text=black},
	RR_12/.style={circle,inner sep=1mm,minimum size=0.5cm,draw,black,very thick,fill=yellow,text=black},
	RR_3/.style={circle,inner sep=1mm,minimum size=0.5cm,draw,black,very thick,fill=cyan,text=black},
	RR_7/.style={circle,inner sep=1mm,minimum size=0.5cm,draw,black,very thick,fill=green,text=black},
	RR_9/.style={circle,inner sep=1mm,minimum size=0.5cm,draw,black,very thick,fill=pink,text=black},
	RR_11/.style={circle,inner sep=1mm,minimum size=0.5cm,draw,black,very thick,fill=gray,text=black},
	 leaf/.style={circle,inner sep=0.3mm,minimum size=0.1cm,draw,very thick,black,fill=white,text=black}]

	\node [ORG_NA] (root) at (0.000000*\xx, 0.500000*\yy)	{};
	\node [ORG_NA] (root-No) at (0.250000*\xx, 0.239130*\yy)	{};
	\node [ORG_NA] (root-Yes) at (0.250000*\xx, 0.760870*\yy)	{};
	\node [INPD_1] (root-No-0) at (0.500000*\xx, 0.065217*\yy)	{};
	\node [EMPUD_2] (root-No-<10) at (0.500000*\xx, 0.239130*\yy)	{};
	\node [EMPUD_2] (root-No->10) at (0.500000*\xx, 0.413043*\yy)	{};
	\node [INPD_3] (root-Yes-0) at (0.500000*\xx, 0.586957*\yy)	{};
	\node [INPD_4] (root-Yes-<10) at (0.500000*\xx, 0.760870*\yy)	{};
	\node [INPD_6] (root-Yes->10) at (0.500000*\xx, 0.934783*\yy)	{};
 	\node [leaf] (root-No-0-No) at (0.750000*\xx, 0.021739*\yy)	{};
    \node [leaf] (root-No-0-ok) at (0.750000*\xx, 0.065217*\yy)	{};
	\node [leaf] (root-No-0-Yes) at (0.750000*\xx, 0.108696*\yy)	{};
	\node [leaf] (root-No-<10-No) at (0.750000*\xx, 0.195652*\yy)	{};
 \node [leaf] (root-No-<10-ok) at (0.750000*\xx, 0.239130*\yy)	{};
	\node [leaf] (root-No-<10-Yes) at (0.750000*\xx, 0.282609*\yy)	{};
	\node [leaf] (root-No->10-No) at (0.750000*\xx, 0.369565*\yy)	{};
 \node [leaf] (root-No->10-ok) at (0.750000*\xx, 0.413043*\yy)	{};
	\node [leaf] (root-No->10-Yes) at (0.750000*\xx, 0.456522*\yy)	{};
	\node [leaf] (root-Yes-0-No) at (0.750000*\xx, 0.543478*\yy)	{};
 \node [leaf] (root-Yes-0-ok) at (0.750000*\xx, 0.586957*\yy)	{};
	\node [leaf] (root-Yes-0-Yes) at (0.750000*\xx, 0.630435*\yy)	{};
	\node [leaf] (root-Yes-<10-No) at (0.750000*\xx, 0.717391*\yy)	{};
 \node [leaf] (root-Yes-<10-ok) at (0.750000*\xx, 0.7608670*\yy)	{};
	\node [leaf] (root-Yes-<10-Yes) at (0.750000*\xx, 0.804348*\yy)	{};
	\node [leaf] (root-Yes->10-No) at (0.750000*\xx, 0.891304*\yy)	{};
 	\node [leaf] (root-Yes->10-ok) at (0.750000*\xx, 0.934783*\yy)	{};
	\node [leaf] (root-Yes->10-Yes) at (0.750000*\xx, 0.978261*\yy)	{};

	\draw[->] (root) -- node [sloped,below]{\small{\texttt{GP} = No}} (root-No);
	\draw[->] (root) -- node [sloped]{\small{\texttt{GP} = Yes}} (root-Yes);
	\draw[->] (root-No) -- node [sloped,below]{\small{\texttt{EMP12} = 10-49}} (root-No-0);
	\draw[->] (root-No) -- node [sloped]{\small{\texttt{EMP12} = 50-249}} (root-No-<10);
	\draw[->] (root-No) -- node [sloped]{\small{\texttt{EMP12} $>$250}} (root-No->10);
	\draw[->] (root-Yes) -- node [sloped,below]{\small{\texttt{EMP12} = 10-49}} (root-Yes-0);
	\draw[->] (root-Yes) -- node [sloped]{\small{\texttt{EMP12} = 50-249}} (root-Yes-<10);
	\draw[->] (root-Yes) -- node [sloped]{\small{\texttt{EMP12} $>$250}} (root-Yes->10);
	\draw[->] (root-No-0) -- node [sloped,below,pos=0.75]{\small{\texttt{EMPUD} = 0\%}} (root-No-0-No);
 \draw[->] (root-No-0) -- node [sloped,pos=0.75]{\small{\texttt{EMPUD} $<$10\%}} (root-No-0-ok);

	\draw[->] (root-No-0) -- node [sloped,pos=0.75]{\small{\texttt{EMPUD} = Yes}} (root-No-0-Yes);
	\draw[->] (root-No-<10) -- node [sloped,below,pos=0.75]{\small{\texttt{EMPUD} = 0\%}} (root-No-<10-No);
 \draw[->] (root-No-<10) -- node [sloped,above,pos=0.75]{\small{\texttt{EMPUD} $<$10\%}} (root-No-<10-ok);
	\draw[->] (root-No-<10) -- node [sloped,pos=0.75]{\small{\texttt{EMPUD} $>$10\%}} (root-No-<10-Yes);
	\draw[->] (root-No->10) -- node [sloped,below,pos=0.75]{\small{\texttt{EMPUD} = 0\%}} (root-No->10-No);
 \draw[->] (root-No->10) -- node [sloped,pos=0.75]{\small{\texttt{EMPUD} $<$10\%}} (root-No->10-ok);
	\draw[->] (root-No->10) -- node [sloped,pos=0.75]{\small{\texttt{EMPUD} $>$10\%}} (root-No->10-Yes);
	\draw[->] (root-Yes-0) -- node [sloped,below,pos=0.75]{\small{\texttt{EMPUD} = 0\%}} (root-Yes-0-No);
 	\draw[->] (root-Yes-0) -- node [sloped,pos=0.75]{\small{\texttt{EMPUD} $<$10\%}} (root-Yes-0-ok);
	\draw[->] (root-Yes-0) -- node [sloped,pos=0.75]{\small{\texttt{EMPUD} $>$10\%}} (root-Yes-0-Yes);
	\draw[->] (root-Yes-<10) -- node [sloped,below,pos=0.75]{\small{\texttt{EMPUD} = 0\%}} (root-Yes-<10-No);
 	\draw[->] (root-Yes-<10) -- node [sloped,pos=0.75]{\small{\texttt{EMPUD} $<$10\%}} (root-Yes-<10-ok);
	\draw[->] (root-Yes-<10) -- node [sloped,pos=0.75]{\small{\texttt{EMPUD} $>$10\%}} (root-Yes-<10-Yes);
	\draw[->] (root-Yes->10) -- node [sloped,below,pos=0.75]{\small{\texttt{EMPUD} = 0\%}} (root-Yes->10-No);
 	\draw[->] (root-Yes->10) -- node [sloped,pos=0.75]{\small{\texttt{EMPUD} $<$10\%}} (root-Yes->10-ok);
	\draw[->] (root-Yes->10) -- node [sloped,pos=0.75]{\small{\texttt{EMPUD} $>$10\%}} (root-Yes->10-Yes);
\end{tikzpicture}}
    \caption{Staged tree for the variable \texttt{EMPUD} conditional on its parents in the \texttt{istat} BN.}
    \label{fig:enter-label}
\end{figure}
\end{document}